\documentclass[12pt]{article}

\usepackage{tikz,tikz-3dplot}
\usetikzlibrary{matrix,arrows}

\usepackage{a4wide,soul,todonotes}

\usepackage{amsmath,amssymb,amsfonts,amsthm}
\newcommand{\bbLbrack}{[\kern-0.4em{[}\,}
\newcommand{\bbRbrack}{\,]\kern-0.4em{]}}
\usepackage[T1]{fontenc}
\usepackage[utf8]{inputenc}
\usepackage{lmodern}
\usepackage{yfonts}

\usepackage{leftidx,extarrows}
\usepackage{undertilde}

\usepackage{mathrsfs,latexsym}
\usepackage[mathscr]{eucal}
\usepackage{mathrsfs}

\newcommand{\II}{{\boldsymbol{1}}}

\newcommand{\CC}{{\mathbb C}}

\newcommand{\RR}{{\mathbb R}}

\newcommand{\NN}{{\mathbb N}}

\newcommand{\BB}{{\mathscr B}}

\newcommand{\HH}{{\mathscr H}}

\newcommand{\KK}{{\mathscr K}}

\newcommand{\Ac}{{\mathcal A}}
\newcommand{\Bc}{{\mathcal B}}

\newcommand{\OO}{{\mathscr O}}

\newcommand{\xb}{{\boldsymbol{x}}}

\newcommand{\ogth}{{\mathfrak o}}
\newcommand{\tgth}{{\mathfrak t}}
\newcommand{\wgth}{{\mathfrak w}}

\DeclareMathOperator{\diam}{diam}

\DeclareMathOperator{\Tr}{Tr}

\newcommand{\ip}[2]{{\langle #1\mid #2\rangle}}

\newcommand{\Fb}{{\boldsymbol{F}}}
\newcommand{\Ib}{{\boldsymbol{I}}}

\newcommand{\Mb}{{\boldsymbol{M}}}
\newcommand{\Nb}{{\boldsymbol{N}}}
\newcommand{\Pb}{{\boldsymbol{P}}}

\newcommand{\Mc}{{\mathcal{M}}}

\newcommand{\Sc}{{\mathcal{S}}}

\newcommand{\Loc}{{\sf Loc}}

\newcommand{\Alg}{{\sf Alg}}

\newcommand{\CAlgSts}{{\sf C^*\hbox{-}AlgSts}}
\newcommand{\CAlg}{{\sf C^*\hbox{-}Alg}}

\newcommand{\Af}{{\mathscr A}}

\newcommand{\Sf}{{\mathscr S}}

\newcommand{\Tc}{{\mathcal T}}

\newcommand{\Xf}{{\mathscr X}}

\newcommand{\id}{{\rm id}}

\newcommand{\Aut}{{\rm Aut}}

\newcommand{\kin}{{\rm kin}}

\DeclareMathOperator{\Diff}{Diff}

\newcounter{tightenum}
\newenvironment{tightitemize}%
{\begin{list}{$\bullet$}{\setlength{\itemsep}{0pt}\setlength{\parsep}{0pt}\setlength{\topsep}{0pt}}}%
{\end{list}}
{\begin{list}{(\roman{tightenum})}{\usecounter{tightenum} \setlength{\itemsep}{0pt}\setlength{\parsep}{0pt}\setlength{\topsep}{0pt}}}%
{\end{list}}

\newcommand{\sS}{\leftidx{_*}{}{S}}
\newcommand{\sT}{\leftidx{_*}{}{T}}

\newcommand{\Ngth}{{\mathfrak N}}
\newcommand{\Rgth}{{\mathfrak R}}

\newtheorem{theorem}{Theorem}[section]
\newtheorem{proposition}[theorem]{Proposition}
\newtheorem{definition}[theorem]{Definition}
\newtheorem{lemma}[theorem]{Lemma}
\newtheorem{corollary}[theorem]{Corollary}
\newtheorem{remark}[theorem]{Remark}

\begin{document} 


\title{The split property for locally covariant quantum field theories
in curved spacetime}
\author{Christopher J Fewster\thanks{\tt chris.fewster@york.ac.uk}\\ Department of Mathematics,
                 University of York, \\
                 Heslington,
                 York YO10 5DD, U.K.}
\date{\today}

\maketitle

\begin{abstract}
The split property expresses the way in which  
local regions of spacetime define subsystems of a quantum
field theory. It is known to hold for general
theories in Minkowski space under the hypothesis of nuclearity. 
Here, the split property is discussed for general locally covariant
quantum field theories in arbitrary globally hyperbolic curved spacetimes,
using a spacetime deformation argument to transport the
split property from one spacetime to another.  It is also
shown how states obeying both the split and (partial) Reeh--Schlieder
properties can be constructed, providing standard split
inclusions of certain local von Neumann algebras. Sufficient conditions are given for the theory to admit such states in ultrastatic spacetimes, from which the general case follows.  A number of consequences
are described, including the existence of local generators for
global gauge transformations, and the classification of certain local von Neumann algebras. Similar arguments are applied to the distal
split property and circumstances are exhibited under which distal splitting
implies the full split property. 
\\[5pt]
{\em Dedicated to the memory of John E Roberts}

\end{abstract}

\section{Introduction}

In relativistic physics, one expects that spacelike separated local
spacetime regions should constitute independent subsystems.
The simplest expression of this in quantum field theory (QFT) is \emph{Einstein causality}, which requires that observables localized in spacelike separated regions  commute and are therefore
commensurable.  Algebraic quantum field theory~\cite{Haag}
offers various strengthened criteria for statistical independence of
observables at spacelike separation (see~\cite{Summers:1990,Summers:2009} 
for reviews) of which the \emph{split property} has turned out to be 
particularly deep and fruitful. For the most part, the split property
has been studied in Minkowski space, while in curved spacetime
results have related to particular linear field theories~\cite{Verch_nucspldua:1993,
DAnHol:2006}. In this paper we establish the split property
in general globally hyperbolic spacetimes, within the framework of locally
covariant QFT~\cite{BrFrVe03} and subject to additional conditions
described below. 

To set the scene, we briefly recall the definition of the split property
in Minkowski space.  In the algebraic framework~\cite{Haag}
one considers a net of $C^*$-algebras $\Ac(O)$ indexed by open bounded regions of Minkowski space. These algebras share a common unit, and
(among other axioms) are isotonous, i.e., $O_1\subset O_2$ implies that $\Ac(O_1)\subset \Ac(O_2)$.
Let $\omega$ be a state on the $C^*$-algebra $\Ac$ generated by all the $\Ac(O)$,
thereby inducing a GNS representation $\pi$ of $\Ac$ on Hilbert space $\HH$
with GNS vector $\Omega$. In this representation we may form local von Neumann algebras
by taking double commutants, $\Rgth(O)=\pi(\Ac(O))''$. Clearly, whenever
$O_1\subset O_2$, there is an inclusion $\Rgth(O_1)\subset \Rgth(O_2)$ of
von Neumann algebras;  following~\cite{DopLon:1984}, the inclusion is said to \emph{split} 
if there is a type $\text{I}$ von Neumann factor $\Ngth$ such that $\Rgth(O_1)\subset\Ngth\subset\Rgth(O_2)$. That is, $\Ngth$ has trivial centre, and is isomorphic as a von Neumann algebra to the algebra
of all bounded operators on some (not necessarily separable) Hilbert space~\cite[Prop.~2.7.19]{BratRob}. 
The state $\omega$ is said to have the \emph{split property} if such inclusions split 
 for all relatively compact $O_1,O_2$ with $\overline{O_1}\subset O_2$. 

The relationship with statistical independence arises as follows. Suppose the
net of local algebras obeys Einstein causality, so that algebras of causally disjoint
regions commute elementwise. If $O_2$ and $O_3$ are causally disjoint
and the inclusion $\Rgth(O_1)\subset\Ngth\subset\Rgth(O_2)$
is split for some $O_1\subset O_2$, then $\Rgth(O_1)$ and $\Rgth(O_3)$ enjoy a high
degree of statistical independence: the algebra they generate is isomorphic to their
$W^*$-tensor product, and thus any normal states $\varphi_1$ and $\varphi_3$ on $\Rgth(O_1)$ and $\Rgth(O_3)$ can be extended
to a normal product state $\varphi$ obeying $\varphi(A_1A_3)=\varphi_1(A_1)\varphi_3(A_3)$ for $A_i\in\Rgth(O_i)$ ($i=1,3$). 

Originally conjectured by Borchers, the split property was first proved for
free fields by Buchholz~\cite{Buc:1974}. Subsequently, it was established
for general models \cite{BucDAnFre:1987}
under suitable hypotheses of \emph{nuclearity}, which controls
the growth of the localized state space with energy. 
As the nuclearity criterion is closely linked to the thermodynamic properties
of the theory~\cite{BucWic:1986,BucJun:1986,BucJun:1989}, it is expected
to hold for many theories of physical interest. In particular, it is  
satisfied by free fields and even for countably many free fields provided
that the spectrum of masses obeys suitable conditions~\cite{BucWic:1986}. 

Our approach to the split property in curved spacetimes is similar in spirit to Sanders' work on the Reeh--Schlieder
property~\cite{Sanders_ReehSchlieder}: the existence of a state with 
the desired properties on the given spacetime is deduced by deforming to a spacetime on which such a state is known (or assumed) to exist. (In the Reeh--Schlieder case, the states obtained are not generally
cyclic for \emph{all} local algebras, and so what is proved is
a partial Reeh--Schlieder property.)
For linear fields, related arguments appear in~\cite{Verch:1993,Dappiaggi:2011} and are also used in the proof of the split property~\cite{Verch_nucspldua:1993,DAnHol:2006}. 
In these cases, the existence of states with the Reeh--Schlieder or split property was proved for ultrastatic spacetimes and used to deduce 
similar results in more general spacetimes.  
A novelty of our specific approach is that we rephrase the deformation arguments for the split and partial Reeh--Schlieder properties into a common language, allowing streamlined proofs running in close analogy to one another. Indeed, we will give a combined result on states obeying both the split and partial Reeh--Schlieder properties, 
thus yielding \emph{standard split inclusions}~\cite{DopLon:1984}.

The paper is structured as follows: in section~\ref{sect:prelim} we describe
the relevant geometrical background, in particular introducing the concept
of a \emph{regular Cauchy pair}, and also recall the main ideas
needed from locally covariant QFT~\cite{BrFrVe03}. Section~\ref{sect:split} contains our main results on the split and Reeh--Schlieder properties. In the latter case, this reproduces results
from~\cite{Sanders_ReehSchlieder}; the interest here is that the proof runs in close analogy to that of the split property, and that the split and Reeh--Schlieder properties can hold simultaneously. We also show that the \emph{distal split property}, in which one demands split inclusions only for situations
in which the outer region is sufficiently larger than the inner, is also
amenable to deformation arguments; moreover, we give results to 
show that the full split property follows from a suitable  
distal split condition for models which have state spaces compatible
with local quasiequivalence and the timeslice condition. It follows that
models that obey a distal split condition but not the full split property (see, e.g.,~\cite[Thm 4.3]{DAnDopFreLon:1987}) cannot admit such state spaces. 

Section~\ref{sect:ultrastatic} describes sufficient conditions for
the existence of states with the Reeh--Schlieder and split
properties in connected ultrastatic spacetimes. As every connected spacetime in our
category can be deformed to such a spacetime, this 
establishes conditions for our results to hold in generality.  
Nonetheless, our deformation arguments hold even for disconnected
spacetimes and we given an example of a state over a disconnected
spacetime with the (full) Reeh--Schlieder property.

By way of outlook, a number of applications of the split property
are described in Section~\ref{sect:split}.
These include the statistical independence at spacelike separation,
the existence of local generators of global gauge transformations 
(established in the Minkowski space case in~\cite{DopLon:1983})
and the identification of local algebras as the unique hyperfinite type $\text{III}_1$ 
factor, up to a tensor product with an abelian algebra.  
However there are numerous additional directions that can be explored, and
in general the split property brings a much more detailed set of tools to bear on the
general analysis of QFT in curved spacetimes than has so far been available.

\section{Preliminaries}\label{sect:prelim}

\subsection{The category $\Loc$ and spacetime deformation}

Locally covariant quantum field theory~\cite{BrFrVe03} describes
QFT on a category of globally hyperbolic spacetimes $\Loc$. 
Fixing a spacetime dimension $n\ge 2$, 
objects of $\Loc$ are quadruples $\Mb=(\Mc,g,\ogth,\tgth)$
where $\Mc$ is a smooth paracompact orientable nonempty 
$n$-manifold with finitely many connected components, 
$g$ is a smooth time-orientable metric of signature $+-\cdots-$
on $\Mc$, $\ogth$ and $\tgth$ are choices of orientation and
time-orientation respectively,\footnote{The orientation (resp., time-orientation) is conveniently represented as a choice of one of the
connected components of the nowhere-zero smooth $n$-forms (resp.,
$g$-timelike $1$-forms) on $\Mc$.} so that the spacetime
$\Mb$ is globally hyperbolic. That is, $\Mb$ has no closed causal
curves and the intersections $J^+_\Mb(p)\cap J^-_\Mb(q)$
of the causal future of $p$ with the causal past of $q$ is
compact (including the possibility of being empty) for any pair
of points $p,q\in\Mc$.   
A morphism between two objects $\Mb=(\Mc,g,\ogth,\tgth)$
and $\Mb'=(\Mc',g',\ogth',\tgth')$ of $\Loc$
is any smooth embedding $\psi:\Mc\to\Mc'$ 
that is isometric, preserves the (time)orientation (i.e., $\psi^*g'=g$, $\psi^*\ogth'=\ogth$, $\psi^*\tgth'=\tgth$) and has
a causally convex image. If the image contains a Cauchy surface
of $\Mb'$, $\psi$ will be described as a \emph{Cauchy morphism}. 

We will often consider open causally convex subsets 
of $\Mb$ with finitely many mutually causally disjoint components; the
set of all such sets will be denoted $\OO(\Mb)$. 
Suppose $\Mb=(\Mc,g,\ogth,\tgth)\in\Loc$, and that $O\in\OO(\Mb)$ is nonempty.
Then $\Mb|_O:=(O,g|_O,\ogth|_O,\tgth|_O)$, i.e., $O$ regarded as a spacetime
in its own right with the induced metric and causal structures from $\Mb$, is an object of $\Loc$,
and the inclusion map $O\hookrightarrow\Mc$ induces a morphism $\iota_{\Mb;O}:\Mb|_O\to\Mb$.
 
There is a useful canonical form for objects of $\Loc$. 
Objects of the form $(\RR\times\Sigma,g, \tgth\wedge\wgth,\tgth)$
where (a) $(\Sigma,\wgth)$ is an oriented $(n-1)$-manifold,
(b) $dt$ is future-directed according to $\tgth$, where
$t$ is the coordinate corresponding to the first factor of 
the Cartesian product $\RR\times\Sigma$, and (c) the metric splits as
\begin{equation}\label{eq:split_metric}
g = \beta  dt\otimes dt -h_t,
\end{equation}
where $\beta\in C^\infty(\RR\times\Sigma)$ is strictly positive and $t\mapsto h_t$ is a smooth choice of (smooth) Riemannian metrics on $\Sigma$,
are said to be in \emph{standard form}. Every leaf $\{t\}\times\Sigma$
is a smooth spacelike Cauchy surface of the spacetime.
The structure theorem for $\Loc$ \cite[\S 2.1]{FewVer:dynloc_theory} is:
\begin{proposition}\label{prop:BS}
Supposing that $\Mb\in\Loc$, let $\Sigma$ be a 
smooth spacelike Cauchy surface of $\Mb$ with induced orientation $\wgth$, and let $t_*\in\RR$. Then there
is a $\Loc$-object  $\Mb_{\text{st}}=(\RR\times\Sigma,g, \tgth\wedge\wgth,\tgth)$ in standard form 
and an isomorphism $\rho:\Mb_{\text{st}}\to\Mb$ in $\Loc$ such that
each $\{t\}\times\Sigma$ is a smooth spacelike Cauchy surface of $\Mb_{\text{st}}$, and $\rho(t_*,\cdot)$ is the inclusion of $\Sigma$ in $\Mb$.
\end{proposition}
Here, the induced orientation $\wgth$ of the Cauchy surface
$\Sigma$ in $\Mb=(\Mc,g,\ogth,\tgth)$ is the unique orientation
such that $\ogth|_\Sigma=\tgth|_\Sigma\wedge\wgth$. Proposition~\ref{prop:BS} 
is a slight elaboration of results due to 
Bernal and S\'anchez (see particularly, \cite[Thm 1.2]{Bernal:2005qf} and \cite[Thm 2.4]{Bernal:2004gm}), which were previously long-standing folk-theorems.  

We will occasionally make use of a Riemannian metric $k$ associated 
with any smooth spacelike Cauchy surface $\Sigma$ of $\Mb\in\Loc$,
uniquely defined so that $\sqrt{k_\sigma(u,u)}n|_{\iota(\sigma)}+\iota_* u$ is a future-directed null vector for every $u\in T_\sigma\Sigma$, where
$n$ is the future-directed unit normal vector field to $\Sigma$ and
$\iota:\Sigma\to\Mb$ is the inclusion map. 
If $\Mb$ is in standard form and
$\Sigma$ is the hypersurface of constant $t$, then $k=\beta^{-1}h_t$, 
of course. The metric $k$ measures spatial distances in terms of light travel times in the rest frame defined by $n$ and is an instantaneous version of the optical metric defined in static spacetimes~\cite{GibbonsPerry:1978}. Accordingly, we refer to $k$
as the \emph{instantaneous optical metric}.\footnote{In~\cite{GibbonsPerry:1978}, the optical metric is
defined for static spacetimes, on spatial sections orthogonal to a timelike Killing vector: if the spacetime metric takes form~\eqref{eq:split_metric} with $h_t\equiv h$ and $\beta$ independent of $t$, then the optical metric is precisely 
$\beta^{-1}h$, coinciding with our instantaneous optical metric. In these
circumstances the geodesics of the optical metric are precisely the spatial projections of null geodesics in the spacetime; this property is not
generally true of the instantaneous optical metric, which, however, is a more general concept.}  

Methods for deforming one globally hyperbolic spacetime
into another go back to the work of Fulling, Narcowich and Wald~\cite{FullingNarcowichWald}, in which the existence
of Hadamard states on ultrastatic spacetimes was used to deduce
their existence on general globally hyperbolic spacetimes.  
As first recognized in~\cite{Verch01}, the same idea can be used to great effect in 
locally covariant QFT. The fundamental 
spacetime deformation result can be formulated as follows
(see~\cite[Prop.~2.4]{FewVer:dynloc_theory}).
\begin{proposition}\label{prop:Cauchy_chain}
Two spacetimes $\Mb$, $\Nb$ in $\Loc$ have oriented-diffeomorphic Cauchy surfaces
if and only if there exists a chain of Cauchy morphisms in $\Loc$ forming a diagram
\begin{equation} \label{eq:Cauchy_chain}
\Mb\xleftarrow{\alpha} \Pb \xrightarrow{\beta} \Ib \xleftarrow{\gamma} \Fb \xrightarrow{\delta} \Nb.
\end{equation} 
\end{proposition}
\begin{proof}  
For later use, we sketch some details needed in the forward implication; 
see~\cite[Prop.~2.4]{FewVer:dynloc_theory} for the full proof.
Assume without loss that $\Mb$ and $\Nb$ are in standard form with $\Mb=(\RR\times\Sigma,g_1,\ogth,\tgth_1)$ and
$\Nb=(\RR\times\Sigma,g_2,\ogth,\tgth_2)$, where $\ogth=\tgth_1\wedge\wgth=\tgth_2\wedge\wgth$ for some
orientation $\wgth$ of $\Sigma$. 
 
Given any reals $t_1<t_1'<t_2'<t_2$, one may construct a metric
$g$ of the form~\eqref{eq:split_metric}, so that 
\begin{tightitemize}
\item $g=g_1$ on $P=(-\infty,t_1)\times\Sigma$ and $g=g_2$ on $F=(t_2,\infty)\times\Sigma$
\item on $(-\infty,t_2')\times\Sigma$ every $g$-timelike vector is $g_1$-timelike
\item on $(t_1',\infty)\times\Sigma$ every $g$-timelike vector is $g_2$-timelike.
\end{tightitemize}
The idea for constructing such a metric is described in~\cite{FullingNarcowichWald};
the argument is simplified and given in more detail in~\cite[Prop.~2.4]{FewVer:dynloc_theory}. Choosing $\tgth$ so that $dt$ is future-directed, the spacetime $\Ib:=(\RR\times\Sigma,g,\ogth,\tgth)$ is globally hyperbolic, because every inextendible $g$-timelike curve intersects each surface of constant $t$ precisely once. In addition, we set
$\Pb:=\Mb|_{P}$ and $\Fb:=\Nb|_{F}$, whereupon the inclusions of $F$ and $P$ into $\RR\times\Sigma$
induce the required Cauchy morphisms in \eqref{eq:Cauchy_chain}. 
\end{proof}

\subsection{Regular Cauchy pairs}

We will be interested in some particular subsets of Cauchy surfaces defined as follows.  

\begin{definition} Let $\Mb\in\Loc$. A \emph{regular Cauchy pair} $(S,T)$ in $\Mb$
is an ordered pair of subsets of $\Mb$, that are nonempty, open, relatively compact subsets of a common smooth spacelike Cauchy surface of $\Mb$ in which $\overline{T}$ has nonempty complement, and so that $\overline{S}\subset T$.
There is a preorder on regular Cauchy pairs so that $(S_1,T_1)\prec (S_2,T_2)$
if and only if $S_2\subset D_\Mb(S_1)$ and $T_1\subset D_\Mb(T_2)$.\footnote{
The preorder is not a partial order, because $(S_1,T_1)\prec (S_2,T_2)\prec (S_1,T_1)$ implies $D_\Mb(S_1)=D_\Mb(S_2)$ and $D_\Mb(T_1)=D_\Mb(T_2)$, but not necessarily  $S_1=S_2$ and $T_1=T_2$.}
\end{definition}
These conditions ensure that $D_\Mb(S)$ and $D_\Mb(T)$ are 
open and casually convex, and hence elements of $\OO(\Mb)$. 
Here, for any subset $U$ of $\Mb$, $D_\Mb(U)$ denotes the
Cauchy development, consisting of all points $p$ in $\Mb$ with the property that
all inextendible piecewise-smooth causal curves through $p$ intersect $U$.
The preorder $\prec$ is illustrated in Figure~\ref{fig:preorder}.

\begin{figure}
\tdplotsetmaincoords{75}{90}
\pgfmathsetmacro{\rvec}{.8}
\pgfmathsetmacro{\thetavec}{15}
\pgfmathsetmacro{\phivec}{60}
\begin{center}
\begin{tikzpicture}[scale=5,tdplot_main_coords]
\coordinate (O) at (0,0,0);
\coordinate (Q) at (0,0,0.4);
\tdplotdrawarc{(O)}{0.35}{0}{360}{anchor=north}{}
\tdplotdrawarc{(O)}{0.5}{0}{360}{anchor=north}{}
\tdplotdrawarc[dashed]{(O)}{0.6}{0}{360}{anchor=north}{}
\tdplotdrawarc{(Q)}{0.17}{0}{360}{anchor=north}{}
\tdplotdrawarc[dotted]{(Q)}{0.25}{0}{360}{anchor=north}{}
\tdplotdrawarc{(Q)}{0.7}{0}{360}{anchor=north}{}
\draw[dotted] (0,0.35,0) -- (0,0.25,0.4);
\draw[dotted] (0,-0.35,0) -- (0,-0.25,0.4);
\draw[dashed] (0,0.6,0) -- (0,0.7,0.4);
\draw[dashed] (0,-0.6,0) -- (0,-0.7,0.4);
\node at (0,0.0,0.0) {$S_1$};
\node at (0,0.425,0) {$T_1$};
\node at (0,0.55,0.4) {$T_2$};
\node at (0,0.0,0.4) {$S_2$};
\end{tikzpicture}
\end{center}
\caption{\small Regular Cauchy pairs with $(S_1,T_1)\prec (S_2,T_2)$. Dotted (resp., dashed) lines
indicate relevant portions of $D_\Mb(S_1)$ (resp., $D_\Mb(T_2)$).}
\label{fig:preorder}
\end{figure}
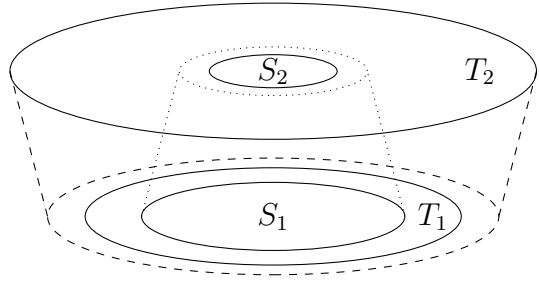
The following lemmas give the properties of regular Cauchy pairs that will be needed. The first is elementary: 
\begin{lemma} \label{lem:Cauchypairs}
Let $\psi:\Mb\to\Nb$ be a Cauchy morphism. Then a pair of subsets $(S,T)$ of $\Mb$ is
a regular Cauchy pair if and only if $(\psi(S),\psi(T))$ is a regular Cauchy pair for $\Nb$ and $\overline{\psi(T)}\subset \psi(\Mb)$. 
\end{lemma}
\begin{proof}
The forward direction holds because the image of a Cauchy surface under a Cauchy morphism
is again a Cauchy surface \cite[Lem.~A.2]{FewVer:dynloc_theory}. 
In the reverse direction, 
similar arguments show that the same is true for pre-images, provided
that the Cauchy surface is completely contained in $\psi(\Mb)$. 
The remaining task is thus to show that $\psi(T)$ lies in at least one
smooth spacelike Cauchy surface contained in $\psi(\Mb)$. Let $\Sigma$ be any smooth spacelike Cauchy surface of $\Nb$ containing $\psi(T)$.
It can happen that $\Sigma$ leaves $\psi(\Mb)$.\footnote{I am grateful to 
Ko Sanders for pointing out this possibility, which was missed in an earlier version.} However, 
as $\overline{\psi(T)}\subset\psi(\Mb)$, there exists a compactly 
supported smooth function $\chi:\Sigma\to\RR$ such that $0\le \chi\le 1$, 
$\chi\equiv 1$ on $\psi(T)$ and $\chi$ is supported on the portion of
$\Sigma$ within $\psi(\Mb)$. Taking any regular value $\alpha\in(0,1)$
of $\chi$, the preimage $\chi^{-1}([\alpha,\infty))$ is a compact submanifold-with-boundary of $\Sigma$ (see, e.g., \cite[\S2]{Milnor_top_diff_view}). Thus, it is also a spacelike and acausal compact codimension-$1$ submanifold-with-boundary of $\psi(\Mb)$ and can
therefore be extended to a smooth spacelike Cauchy surface contained
in $\psi(\Mb)$ \cite[Thm~1.1]{Bernal:2005qf}, which necessarily contains
$\psi(T)$.
\end{proof}  
The next two results indicate the extent to which ordered regular Cauchy pairs may be found in nearby Cauchy surfaces. We use two pieces
of notation: first, when a spacetime is presented in standard form with manifold
$\RR\times\Sigma$, we denote any regular Cauchy pair of the form
$(\{t\}\times S,\{t\}\times T)$ by $(S,T)_t$ for brevity; second, 
the notation $A\Subset B$ indicates that $\overline{A}$ is compact and contained in $B$. 
\begin{lemma}\label{lem:lightspeed}
Suppose that $\Mb$ takes standard form with underlying manifold $\RR\times\Sigma$, 
and that $T$ is an open relatively compact subset of $\Sigma$ with nonempty exterior. 
Let $t_*\in\RR$ and let $B(U,\delta)$ denote the open ball of radius $\delta$ about 
$U\subset \Sigma$ with respect to the instantaneous optical metric on $\{t_*\}\times\Sigma$ induced by $\Mb$. 
For all $\delta>0$ such that $B(T,\delta)$ is relatively compact with nonempty exterior,\footnote{The existence of a relatively compact $\delta$-ball
about $T$ follows from the existence of a compact exhaustion of $\Sigma$~\cite[Prop.~4.76]{Lee:topman}, given that $T$ has nonempty exterior.} there exists $\epsilon>0$ such that $\{t\}\times T
\subset D_\Mb(\{t'\}\times B(T,\delta))$ provided that $t,t'\in(t_*-\epsilon,t_*+\epsilon)$. Further, suppose that
$S\subset\Sigma$ is open and relatively compact
with $B(S,\delta) \Subset T$, then 
$(B(S,\delta),T)_t\prec (S,B(T,\delta))_{t'}$ 
for any $t,t'\in (t_*-\epsilon,t_*+\epsilon)$. 
\end{lemma}
\begin{proof}
Without loss of generality take $t_*=0$ and denote the instantaneous optical metric induced on 
$\Sigma$ via the slice $\{\tau\}\times\Sigma$ of $\Mb$ by $k_{\tau}$. As
$B(T,\delta)$ is relatively compact, there is a constant $K\ge 1$ such that
$k_{0,\sigma}(u,u) \le K k_{{\tau,\sigma}}(u,u)$ for all $u\in T_{\sigma}\Sigma$, 
 $(\tau,\sigma)\in [0,\delta]\times \overline{B(T,\delta)}$.  We set 
 $\epsilon= \delta/(2\sqrt{K})$ and choose any $t,t'\in(-\epsilon,\epsilon)$. 
 
Any smooth inextendible $\Mb$-causal curve $\gamma$ may be parameterized so that $\gamma(\tau) = (\tau,\sigma(\tau))$  ($\tau\in\RR$), where
$\sigma$ is smooth and necessarily obeys $k_{\tau,\sigma(\tau)}(\dot{\sigma}(\tau),\dot{\sigma}(\tau))\le 1$ for all $\tau\in\RR$. Thus if $\gamma(t)\in T$, and $|t|,|t'|<\epsilon$, we may estimate (using
the $k_0$ metric)
\begin{equation}\label{eq:delta_est}
{\text{dist}}(\sigma(t),\sigma(t')) \le \sqrt{K}|t-t'|  < 2\epsilon \sqrt{K}=\delta 
\end{equation}
and hence $\gamma(t')\in B(T,\delta)$ as required. (We eliminate the possibility
that $\sigma(\tau)$ leaves $B(T,\delta)$ at some intermediate time by similar reasoning.)
Thus  $\{t\}\times T \subset D_\Mb(\{t'\}\times B(T,\delta))$ as required. 
Under the additional assumptions concerning $S$, we may
apply the same estimates, and reverse the roles of $t$ and $t'$ to find
$\{t'\}\times S \subset D_\Mb(\{t\}\times B(S,\delta))$ 
in addition to the previous statement concerning $T$,
whereupon  $(B(S,\delta),T)_t\prec (S,B(T,\delta))_{t'}$.
\end{proof} 

\begin{lemma} \label{lem:step}
Suppose that $\Mb\in\Loc$ takes standard form with underlying manifold $\RR\times\Sigma$. 

(a) Let $S_1,S_2,T_1,T_2$ be open relatively compact subsets of $\Sigma$ with $S_2\Subset   S_1 \Subset  T_1 \Subset T_2$  
and so that $T_2$ has nonempty exterior. 
Then, for any $t_*\in\RR$, there exists $\epsilon>0$ such that
\begin{equation}
(S_{1}, T_{1})_{t_1} \prec (S_{2}, T_{2})_{t_2} 
\end{equation}
for all $t_1,t_2\in (t_*-\epsilon,t_*+\epsilon)$.  

(b) Let 
$(S,T)_t$ be a regular Cauchy pair in $\Mb$
for some $t\in\RR$. 
Let $S_{\text{inner}}$, $S_{\text{outer}}$, $T_{\text{inner}}$ and
$T_{\text{outer}}$ be any open relatively compact subsets of $\Sigma$
such that\footnote{The existence of such sets follows from the
assumptions on $S$ and $T$.}
\begin{equation}\label{eq:STchain}
S_{\text{inner}}\Subset   S \Subset  S_{\text{outer}}
\Subset
T_{\text{inner}}\Subset T \Subset T_{\text{outer}}
\end{equation}
and so that $T_{\text{outer}}$ has nonempty exterior. 
Then there exists an $\epsilon>0$ such that 
\begin{equation}
(S_{\text{outer}}, T_{\text{inner}})_{t'} \prec 
(S,T)_t \prec (S_{\text{inner}}, T_{\text{outer}})_{t'} 
\end{equation}
for all $t'\in (t-\epsilon,t+\epsilon)$. In particular, 
every Cauchy surface $\{t'\}\times\Sigma$ with $|t'-t|<\epsilon$ contains a regular Cauchy pair that precedes $(S,T)_t$ and one that is preceded by it.
\end{lemma}
\begin{proof}
(a) Let $B(U,\delta)$ denote the open $\delta$-ball about $U\subset\Sigma$ in the instantaneous optical metric on $\{t_*\}\times\Sigma$. By assumption on the various sets in the
hypotheses, we may choose $\delta>0$ such 
that $B(S_{2},\delta)\subset S_1$ and $B(T_1,\delta)\subset T_{2}$.  Using these inclusions together with  Lemma~\ref{lem:lightspeed}, there exists
$\epsilon>0$ such that 
\begin{equation}
(S_{1}, T_{1})_{t_1} \prec
(B(S_2,\delta), T_{1})_{t_1} \prec
(S_2, B(T_{1},\delta))_{t_2} \prec (S_2, T_2)_{t_2} 
\end{equation}
holds for all  $t_1,t_2\in (t_*-\epsilon,t_*+\epsilon)$.  
(b) Apply (a) twice, taking $t_*=t$.
\end{proof}

\begin{remark}\label{rem:multiCauchypair}\emph{
It follows immediately that, if finitely many regular Cauchy pairs $(S_j,T_j)$ ($1\le j\le N$) are specified in the Cauchy surface $\{t\}\times\Sigma$, then every Cauchy surface $\{t'\}\times\Sigma$ with 
$t'$ sufficiently close to $t$ contains, for each $j$, a regular Cauchy pair preceding 
$(S_j,T_j)$ and one that is preceded by it.}
\end{remark}

\subsection{Locally covariant quantum field theory}

The basic premise of locally covariant QFT~\cite{BrFrVe03} is that a theory is given by a functor $\Af:\Loc\to\CAlg$, where $\CAlg$ is the category of unital $C^*$-algebras and injective unit-preserving
$*$-homomorphisms.\footnote{Other target categories are possible and frequently employed, for example the category $\Alg$ of unital $*$-algebras with injective unit-preserving $*$-homomorphisms.}
This means that each spacetime $\Mb$ corresponds to a $C^*$-algebra $\Af(\Mb)$, 
and that every morphism $\psi:\Mb\to\Nb$ between spacetimes has a corresponding $\CAlg$-morphism
$\Af(\psi):\Af(\Mb)\to\Af(\Nb)$, subject to the requirement that $\Af(\id_\Mb)=\id_{\Af(\Mb)}$
and $\Af(\psi\circ\varphi)=\Af(\psi)\circ\Af(\varphi)$. 

Given such a functor, a net of local algebras may be defined in each spacetime $\Mb\in\Loc$ 
by setting $\Af^\kin(\Mb;O)$ to be the image of the map $\Af(\iota_{\Mb;O})$ for
each nonempty $O\in\OO(\Mb)$.  As described in~\cite{BrFrVe03}, these local algebras
obey suitable generalizations of the assumptions in the Araki--Haag--Kastler framework~\cite{Haag}. In particular, they are \emph{isotonous}: if $O_1\subset O_2$ then $\Af^\kin(\Mb;O_1)\subset\Af^\kin(\Mb;O_2)$. 

The additional assumptions we will use are that the theory is \emph{Einstein causal}:
if $O_1,O_2\in\OO(\Mb)$ are causally disjoint (in the sense that no
causal curve connects them), then $\Af^\kin(\Mb;O_1)$ and
$\Af^\kin(\Mb;O_2)$ commute, and that the theory has the \emph{timeslice property}:
if $\psi:\Mb\to\Nb$ is Cauchy, then $\Af(\psi)$ is an isomorphism.
\begin{definition} \label{def:lcQFT}
A \emph{locally covariant QFT} is a functor $\Af:\Loc\to\CAlg$
obeying Einstein causality and having the timeslice property. 
\end{definition}

The utility of the deformation result Proposition~\ref{prop:Cauchy_chain}
arises because any chain of Cauchy morphisms such as \eqref{eq:Cauchy_chain} induces, by the timeslice property, an
isomorphism 
\begin{equation}\label{eq:alg_Cauchy_chain}
\Af(\delta)\circ\Af(\gamma)^{-1}\circ\Af(\beta)\circ
\Af(\alpha)^{-1}:\Af(\Mb)\to\Af(\Nb) .
\end{equation}
Although such isomorphisms are not canonical, owing to the many choices used in the construction, they often permit the transfer of properties and
structures between the instantiations of the theory on $\Mb$ and $\Nb$.

The description just given encodes the algebraic aspects of the theory.
To incorporate states as well, we first define a category $\CAlgSts$ as follows.
Objects of $\CAlgSts$ are pairs $(\Ac,\Sc)$, where $\Ac\in\CAlg$ and $\Sc$
is a state space for $\Ac$, i.e., a convex subset of the set of all states on $\Ac$,
that is closed under operations induced by $\Ac$.\footnote{That is, if
$\omega\in\Sc$ and $B\in\Ac$ with $\omega(B^*B)> 0$, then the state $\omega_B(A):=\omega(B^*AB)/\omega(B^*B)$ is also an element of $\Sc$.}
A morphism in $\CAlgSts$ between $(\Ac,\Sc)$ and $(\Bc,\Tc)$ is induced
by any $\CAlg$-morphism $\alpha:\Ac\to\Bc$ such that $\alpha^*\Tc\subset\Sc$;
as a slight abuse of notation we will often denote the $\CAlgSts$-morphism
in the same way as its underlying $\CAlg$ morphism. 

A state space for a locally covariant QFT $\Af:\Loc\to\CAlg$ is
an assignment of state space $\Sf(\Mb)$ to each $\Af(\Mb)$ ($\Mb\in\Loc$)
so that $\Xf(\Mb)=(\Af(\Mb),\Sf(\Mb))$ defines a functor $\Xf:\Loc\to\CAlgSts$
for which each $\Xf(\psi)$ has underlying $\CAlg$-morphism $\Af(\psi)$. 
We say that $\Xf$ obeys the timeslice axiom
if $\Xf(\psi)$ is an isomorphism in $\CAlgSts$ for all Cauchy morphisms $\psi:\Mb\to\Nb$,
which means that $\Af(\psi)^*\Sf(\Nb)=\Sf(\Mb)$
(of course, $\Af(\psi)$ is also an isomorphism because $\Af$ obeys Definition~\ref{def:lcQFT}). In this case $\Xf$ will be described
as a \emph{locally covariant QFT with states}.

\section{Main Results} \label{sect:split}

\subsection{The split property}

The split property is defined as follows.\footnote{This definition directly generalizes
that used in Minkowski space, but differs from the condition studied in~\cite{BrFrImRe:2014}
and discussed briefly at the end of this section.}
\begin{definition}\label{def:split}
Let $\Af:\Loc\to\CAlg$ be a locally covariant QFT and $\Mb\in\Loc$.
A state $\omega$ on $\Af(\Mb)$ is said to have the \emph{split property}
for a regular Cauchy pair $(S,T)$ if, in the GNS representation ($\HH,\pi,\Omega)$ of $\Af(\Mb)$ induced by $\omega$, there is a type $\text{\em I}$ factor $\Ngth$ such that
\begin{equation}\label{eq:split}
\pi(\Af^\kin(\Mb;D_\Mb(S)))''\subset \Ngth \subset  \pi(\Af^\kin(\Mb;D_\Mb(T)))''.
\end{equation} 
(For brevity, we will sometimes say that $\omega$ is split for $(S,T)$.)
\end{definition}
\begin{remark}\label{rem:split}\emph{
If $\omega$ has the split property for $(S,T)$ then 
it does for every $(\tilde{S},\tilde{T})$ with $(S,T)\prec (\tilde{S},\tilde{T})$:
for $\tilde{S}\subset D_\Mb(S)$ implies $D_\Mb(\tilde{S})\subset
D_\Mb(S)$ and hence by isotony
\begin{align}
\pi(\Af^\kin(\Mb;D_\Mb(\tilde{S})))''  &\subset 
\pi(\Af^\kin(\Mb;D_\Mb(S)))'', \nonumber\\
 \pi(\Af^\kin(\Mb;D_\Mb(T)))''  &\subset  \pi(\Af^\kin(\Mb;D_\Mb(\tilde{T})))''.
\end{align}
Moreover, if nonempty $O_i\in\OO(\Mb)$ obey $O_1\subset D_\Mb(S)$, $D_\Mb(T)\subset O_2$, then there is a split inclusion
\begin{equation}\label{eq:splitOi}
\pi(\Af^\kin(\Mb;O_1))''\subset\Ngth \subset \pi(\Af^\kin(\Mb;O_2))''
\end{equation}
by the same argument. }
\end{remark}

\begin{lemma}\label{lem:split}
Suppose $\psi:\Mb\to\Nb$ is a Cauchy morphism and let $\Af$ be a locally covariant QFT. A state $\omega_\Nb$ on $\Af(\Nb)$ has
the split property for a regular Cauchy pair $(\psi(S),\psi(T))$
with $\overline{\psi(T)}\subset \psi(\Mb)$ 
 if and only if $\Af(\psi)^*\omega_\Nb$ has
the split property for $(S,T)$. (As $\Af(\psi)$ is an isomorphism, 
this implies that $\omega_\Mb$ is 
split for $(S,T)$ if and only if $(\Af(\psi)^{-1})^*\omega_\Mb$
is split for $(\psi(S),\psi(T))$.)
\end{lemma}
\begin{proof}
Let $\omega_\Mb=\Af(\psi)^*\omega_\Nb$ 
and write  
$(\HH_{\omega_{\star}},\pi_{\omega_{\star}},\Omega_{\omega_{\star}})$,
where $\star=\Mb$ or $\Nb$, for the corresponding GNS representations. 
As $\Af(\psi)$ is an isomorphism there is a unitary $U:\HH_{\omega_\Mb}\to
\HH_{\omega_\Nb}$ so that $U\Omega_{\omega_\Mb} = \Omega_{\omega_\Nb}$ and
\begin{equation}\label{eq:intertwine}
U \pi_{\omega_\Mb}(A)= \pi_{\omega_\Nb}(\Af(\psi) A)U,\qquad (A\in\Af(\Mb)).
\end{equation} 
Consequently,  $\pi_{\omega_\Nb}(\Af^\kin(\Nb;\psi(O)))'' =
U\pi_{\omega_\Mb}(\Af^\kin(\Mb;O))''U^{-1}$ for any nonempty $O\in\OO(\Mb)$, and 
as $U\Ngth U^{-1}$ is a type $\text{I}$ factor if and only if $\Ngth$ is,
the  result follows.
\end{proof}

We now present our first deformation result on the split property.
\begin{theorem}\label{thm:split}
Suppose $\Af$ is a locally covariant QFT. Let $\Mb,\Nb\in\Loc$ have oriented-diffeomorphic Cauchy surfaces and suppose $\omega_\Nb$ is a state on $\Af(\Nb)$ 
that has the split property for all regular Cauchy pairs in $\Nb$.  
Given any regular Cauchy pair $(S_\Mb,T_\Mb)$ in $\Mb$, there is a chain of Cauchy morphisms
between $\Mb$ and $\Nb$ inducing an isomorphism $\nu:\Af(\Mb)\to\Af(\Nb)$ such
that $\nu^*\omega_\Nb$ has the split property for $(S_\Mb,T_\Mb)$. Consequently (by Remark~\ref{rem:split})
if nonempty $O_i\in\OO(\Mb)$ are such that $O_1\subset D_\Mb(S_\Mb)$, $D_\Mb(T_\Mb)\subset O_2$, then there is a split inclusion of the form
\eqref{eq:splitOi} in the GNS representation of $\nu^*\omega_\Nb$.  
\end{theorem}
\begin{proof}
Assume without loss of generality (by Proposition~\ref{prop:BS} and Lemma~\ref{lem:Cauchypairs}) that $\Mb$ is in standard form
$\Mb=(\RR\times\Sigma,g_\Mb,\ogth,\tgth_\Mb)$ and that 
$S_\Mb$ and $T_\Mb$ lie in the Cauchy surface $\{t_\Mb\}\times\Sigma$ for some $t_\Mb\in\RR$.  
By Lemma~\ref{lem:step} there exist $t_*>t_\Mb$ and a regular Cauchy pair $(S_*,T_*)$ in $\{t_*\}\times\Sigma$
such that $(S_*,T_*)\prec_\Mb (S_\Mb,T_\Mb)$, where $\prec_\Mb$
indicates the preorder with respect to the causal structure of $\Mb$. 

Now we may also assume without loss of generality that $\Nb$ is also in standard form $\Nb=(\RR\times\Sigma,g_\Nb,\ogth,\tgth_\Nb)$.
As $(S_*,T_*)$ is also a regular Cauchy pair for $\Nb$,
there exist $t_\Nb>t_*$ and a regular Cauchy pair $(S_\Nb,T_\Nb)$ in $\{t_\Nb\}\times\Sigma$ such that $(S_\Nb,T_\Nb)\prec_\Nb (S_*,T_*)$.

We now construct a metric $g$ using Prop.~\ref{prop:Cauchy_chain},
choosing the values $t_1,t_1',t_2',t_2$ so that $t_\Mb<t_1<t_1'<t_*<t_2'<t_2<t_\Nb$, 
and thus creating an interpolating
globally hyperbolic spacetime $\Ib$ and a chain of Cauchy morphisms \eqref{eq:Cauchy_chain}.
The key point is that  $(S_\Nb,T_\Nb) \prec_\Ib(S_*,T_*) $ and
$(S_*,T_*)\prec_\Ib (S_\Mb,T_\Mb)$  and hence
$(S_\Nb,T_\Nb)\prec_\Ib(S_\Mb,T_\Mb)$. 
To see this, consider any inextendible
$g$-timelike curve $\gamma$ through $S_\Mb$. In the region $t\le t_*$ this is
also a $g_\Mb$-timelike curve and intersects $S_*$, because
$S_\Mb\subset D_\Mb(S_*)$. Thus $S_\Mb\subset D_\Ib(S_*)$. Similarly,
if $\gamma$ is an inextendible $g$-timelike curve through $T_*$, it is 
$g_\Mb$-timelike in $\Mb$ and intersects $T_\Mb$, so $S_*\subset D_\Ib(S_\Mb)$. This shows that
$(S_*,T_*)\prec_\Ib (S_\Mb,T_\Mb)$; one proves $(S_\Nb,T_\Nb)\prec_\Ib (S_*,T_*)$  
in the same way. 

As $\omega_\Nb$ has the split property for $(S_\Nb,T_\Nb)$ in $\Nb$,
it follows (applying Lemma~\ref{lem:split} twice) that $(\Af(\delta)\circ\Af(\gamma)^{-1})^*\omega_\Nb$ has the
split property for $(S_\Nb,T_\Nb)$, as a regular Cauchy pair in $\Ib$,
and hence for $(S_\Mb,T_\Mb)$, again as a regular Cauchy pair in $\Ib$,
because $(S_\Nb,T_\Nb)\prec_\Ib (S_\Mb,T_\Mb)$.
Two further applications of Lemma~\ref{lem:split} show that
$(\Af(\beta)\circ\Af(\alpha)^{-1})^*
(\Af(\delta)\circ\Af(\gamma)^{-1})^*\omega_\Nb = \nu^*\omega_\Nb$
has the split property for $(S_\Mb,T_\Mb)$ in $\Mb$. 
\end{proof}
\begin{remark}\label{rem:multisplit}\emph{
The result may be extended as follows. Suppose finitely many regular
Cauchy pairs $(S_\Mb^{(j)},T_\Mb^{(j)})$ ($1\le j\le N$), lying in a common Cauchy surface of $\Mb$ are given. Owing to Remark~\ref{rem:multiCauchypair}, the values $t_*$ and $t_\Nb$ 
in the proof above may be chosen so that there are Cauchy pairs 
$(S_*^{(j)},T_*^{(j)})$ and $(S_\Nb^{(j)},T_\Nb^{(j)})$ ($1\le j\le N$) lying
in the hypersurfaces $\{t_*\}\times\Sigma$ and $\{t_\Nb\}\times\Sigma$ respectively
so that  $(S_\Nb^{(j)},T_\Nb^{(j)})\prec_\Nb (S_*^{(j)},T_*^{(j)}) \prec_\Mb (S_\Mb^{(j)},T_\Mb^{(j)})$ for each $1\le j\le N$ and hence 
 $(S_\Nb^{(j)},T_\Nb^{(j)})\prec_\Ib (S_\Mb^{(j)},T_\Mb^{(j)})$ for a common interpolating
metric. Then the state $\nu^*\omega_\Nb$ has the split property 
for each of the pairs $(S_\Mb^{(j)},T_\Mb^{(j)})$ ($1\le j\le N$).}
\end{remark}

For theories with states $\Xf=(\Af,\Sf):\Loc\to\CAlgSts$, we may say a little more.
First, if the state $\omega_\Nb$ in the hypotheses of Theorem~\ref{thm:split}
belongs to the state space $\Sf(\Nb)$, then the induced state obeys $\nu^*\omega_\Nb\in\Sf(\Mb)$,
as a result of the timeslice property for $\Xf$ and the fact that $\nu$ arises from 
a chain of Cauchy morphisms. Much more follows if each $\Sf(\Mb)$ consists
of mutually \emph{locally quasi-equivalent} states on $\Af(\Mb)$,
in which case we describe $\Xf$ as obeying local quasi-equivalence.
This condition requires that  
for every spacetime $\Mb$, relatively compact nonempty $O\in\OO(\Mb)$ and states $\omega_i\in\Sf(\Mb)$ ($i=1,2$), the
GNS representations $(\HH_{\omega_i},\pi_{\omega_i},\Omega_i)$ restrict to quasi-equivalent 
representations of $\Af^\kin(\Mb;O)$, i.e.,
there is an isomorphism of von Neumann algebras
$\beta:\pi_{\omega_1}(\Af^\kin(\Mb;O))'' \to\pi_{\omega_2}(\Af^\kin(\Mb;O))''$
such that $\beta\circ\pi_{\omega_1}(A)=\pi_{\omega_2}(A)$ for all $A\in 
\Af^\kin(\Mb;O)$.\footnote{An equivalent definition of quasi-equivalence  is that the sets of states on $\Af^\kin(\Mb;O)$ induced by
density matrices on $\HH_1$ and $\HH_2$ coincide \cite[Thm 2.4.26]{BratRob}.}
An example of a locally quasi-equivalent state space ~\cite[Thm 3.4]{BrFrVe03} is provided by 
taking, in each spacetime $\Mb$, all states on the Weyl algebra of the Klein--Gordon field that are locally quasi-equivalent to any quasi-free
Hadamard state (the latter being mutually locally quasi-equivalent~\cite{Verch:1994}).  We have:
\begin{lemma} \label{lem:quasi}
If state $\omega_1$ has the split property for regular Cauchy pair $(S,T)$
in $\Mb$ and $\omega_2$ is locally quasi-equivalent to $\omega_1$, then $\omega_2$
also has the split property for $(S,T)$. 
\end{lemma}
\begin{proof}
Let $\Ngth$ be the type $\text{I}$ factor 
obeying \eqref{eq:split} and  let $\beta:\pi_{\omega_1}(\Af^\kin(\Mb;D_\Mb(T)))'' \to\pi_{\omega_2}(\Af^\kin(\Mb;D_\Mb(T)))''$
be the isomorphism induced by local quasi-equivalence, obeying 
$\beta\circ\pi_{\omega_1}=\pi_{\omega_2}$ on $\Af^\kin(\Mb;D_\Mb(T))$.
In particular, $\beta$ restricts to an isomorphism of 
$\pi_{\omega_1}(\Af^\kin(\Mb;D_\Mb(S)))'' \to\pi_{\omega_2}(\Af^\kin(\Mb;D_\Mb(S)))''$.
Then $\beta(\Ngth)$ is a type $\text{I}$ factor, and clearly obeys
$\pi_{\omega_2}(\Af^\kin(\Mb;D_\Mb(S)))'' \subset\beta(\Ngth)\subset \pi_{\omega_2}(\Af^\kin(\Mb;D_\Mb(T)))''$.
\end{proof}
As an immediate consequence (just as was argued for the Klein--Gordon theory in~\cite{Verch_nucspldua:1993}):
\begin{theorem}
Suppose $\Xf=(\Af,\Sf):\Loc\to\CAlgSts$ is a locally covariant QFT with states obeying local quasi-equivalence. Let $\Mb,\Nb\in\Loc$ have oriented-diffeomorphic Cauchy surfaces and suppose $\omega_\Nb\in\Sf(\Nb)$  has the split property for all regular Cauchy pairs in $\Nb$.  
Then every state $\omega_\Mb\in\Sf(\Mb)$ obeys the split property for all regular
Cauchy pairs in $\Mb$.  Consequently, 
if $O_i\in\OO(\Mb)$ are such that $O_1\subset D_\Mb(S)$, $D_\Mb(T)\subset O_2$,
for a regular Cauchy pair $(S,T)$ in $\Mb$, then there is a split inclusion of the form
\eqref{eq:splitOi} in the GNS representation induced by any state of $\Sf(\Mb)$.   
\end{theorem}
\begin{proof}
For each regular Cauchy pair $(S_\Mb,T_\Mb)$ of $\Mb$,  
Theorem~\ref{thm:split} shows the existence of some state in $\Sf(\Mb)$ having
the split property for $(S_\Mb,T_\Mb)$, and hence by Lemma~\ref{lem:quasi}
and local quasi-equivalence of $\Xf$, the same is true for all states of $\Sf(\Mb)$.
\end{proof}

\subsection{Partial Reeh--Schlieder results}

As already mentioned, our result on the split property
was inspired by Sanders' partial analogue of the Reeh--Schlieder theorem~\cite{Sanders_ReehSchlieder}. The original
Reeh--Schlieder theorem~\cite{ReehSchlieder:1961} establishes that the Minkowski vacuum vector is cyclic for all local algebras, and consequently
separating for all local algebras for regions with nonempty causal complement. The results of~\cite{Sanders_ReehSchlieder} 
demonstrate the existence of states with partial Reeh--Schlieder
properties:  given a spacetime region in $\Mb$, one may find (suitably regular) states that are cyclic for the corresponding local algebra, on the
assumption that $\Mb$ can be deformed to a spacetime that admits a (suitably regular) state enjoying the full Reeh--Schlieder
property of being cyclic for all local algebras. 

The introduction of regular Cauchy pairs allows for a streamlined proof of Sanders' result, which we give for completeness. More significantly, we combine this proof with that of our result on the split property to demonstrate the existence of
states obeying both the split and Reeh--Schlieder properties, which give
so-called standard split inclusions~\cite{DopLon:1984}.

The properties we will consider are given as follows. Terminology 
differs from~\cite{Sanders_ReehSchlieder}.
\begin{definition}
Let $\Af:\Loc\to\CAlg$ be a locally covariant QFT and $\Mb\in\Loc$.
A state $\omega$ on $\Af(\Mb)$ is said to have the \emph{Reeh--Schlieder property}
for a regular Cauchy pair $(S,T)$ if, in the GNS representation 
($\HH,\pi,\Omega)$ of $\Af(\Mb)$ induced by
$\omega$, the GNS vector $\Omega$ is cyclic for $\pi(\Af^\kin(\Mb;D_\Mb(S)))''$ and separating
for $\pi(\Af^\kin(\Mb;D_\Mb(T)))''$. 
For brevity, we will sometimes say that \emph{$\omega$ is Reeh--Schlieder for $(S,T)$}. If $O\in\OO(\Mb)$ and $\Omega$ 
is both cyclic and separating for $\pi(\Af^\kin(\Mb;O))''$, 
we will say that $\omega$ \emph{has the Reeh--Schlieder property for $O$}.\footnote{Sanders~\cite{Sanders_ReehSchlieder} uses this
term for cyclicity alone.} 
\end{definition}
Note that we regard the separation condition as part of the Reeh--Schlieder property, which turns out to expedite the proofs below. See Corollary~\ref{cor:RS} for a formulation involving only cyclicity as a hypothesis. 
\begin{remark}\label{rem:RS}
\emph{If a vector is separating for an algebra, it is separating for any subalgebra
thereof; if it is cyclic for an algebra, it is cyclic for any algebra of which it is a subalgebra. 
Thus it is clear that if $\omega$ has the Reeh--Schlieder property for $(S,T)$ then 
it does for every $(\tilde{S},\tilde{T})$ with $(\tilde{S},\tilde{T})\prec (S,T)$.\footnote{Note the reversal of order relative to Remark~\ref{rem:split}.} Moreover, if
$O\in\OO(\Mb)$ is such that $D_\Mb(S)\subset O\subset D_\Mb(T)$, then the
GNS vector of $\omega$ is both cyclic and separating for $\pi(\Af^\kin(\Mb;O))''$, i.e.,  $\omega$ is Reeh--Schlieder for $O$.
Note that the separating property is defined at the level of the represented
algebras. If $\omega$ induces a faithful GNS representation, we would have
the stronger property that $\omega(A^*A)=0$ for $A\in\Af^\kin(\Mb;O)$ 
implies $A=0$. }
\end{remark} 
 
\begin{lemma}\label{lem:RS}
Let $\Af$ be a locally covariant QFT. Let $(S,T)$ be a regular
Cauchy pair in $\Mb\in\Loc$ and suppose 
$\psi:\Mb\to\Nb$ is Cauchy.  A state $\omega_\Nb$ on $\Af(\Nb)$ is Reeh--Schlieder for a regular Cauchy pair $(\psi(S),\psi(T))$ with $\overline{\psi(T)}\subset \psi(\Mb)$
if and only if $\Af(\psi)^*\omega_\Nb$ is
Reeh--Schlieder for $(S,T)$. (As $\Af(\psi)$ is an isomorphism, 
this implies that $\omega_\Mb$ is 
Reeh--Schlieder for $(S,T)$ if and only if $(\Af(\psi)^{-1})^*\omega_\Mb$
is Reeh--Schlieder for $(\psi(S),\psi(T))$.)
\end{lemma}
\begin{proof}
As in the proof of Lemma~\ref{lem:split}, we set $\omega_\Mb=\Af(\psi)^*\omega_\Nb$,
and infer the existence of a unitary $U:\HH_{\omega_\Mb}\to
\HH_{\omega_\Nb}$ so that $U\Omega_{\omega_\Mb} = \Omega_{\omega_\Nb}$ and
$\pi_{\omega_\Nb}(\Af^\kin(\Nb;\psi(O)))'' =
U\pi_{\omega_\Mb}(\Af^\kin(\Mb;O))''U^{-1}$ for
$O\in\OO(\Mb)$. Consequently, $\Omega_{\omega_\Nb}$ is cyclic (resp., separating) for 
$\pi_{\omega_\Nb}(\Af^\kin(\Nb;\psi(O)))''$ if and only if $\Omega_{\omega_\Mb}$ 
is cyclic (resp., separating) for $\pi_{\omega_\Mb}(\Af^\kin(\Mb;O))''$. 
\end{proof}

An analogue of Theorem~\ref{thm:split} now gives a partial Reeh--Schlieder result. 
\begin{theorem}\label{thm:RS}
Suppose $\Af$ is a locally covariant QFT. Let $\Mb,\Nb\in\Loc$ have oriented-diffeomorphic Cauchy surfaces and suppose $\omega_\Nb$ is a state on $\Af(\Nb)$ 
that is Reeh--Schlieder for all regular Cauchy pairs.  
Given any regular Cauchy pair $(S_\Mb,T_\Mb)$ in $\Mb$, there is a chain of Cauchy morphisms
between $\Mb$ and $\Nb$ inducing an isomorphism $\nu:\Af(\Mb)\to\Af(\Nb)$ such
that $\nu^*\omega_\Nb$ has the Reeh--Schlieder property for $(S_\Mb,T_\Mb)$. Consequently, if $O\in\OO(\Mb)$ is relatively compact with nontrivial causal complement $O':=\Mb\setminus\overline{J_\Mb(O)}$, there is a state (formed in the same way) on $\Af(\Mb)$ with the Reeh--Schlieder property for $O$.
\end{theorem}
\begin{proof}
The first part of the argument is identical to that of Theorem~\ref{thm:split}, except
that we replace $\prec$ by $\succ$, and `split' by `Reeh--Schlieder' on every occasion, and use Lemma~\ref{lem:RS} and Remark~\ref{rem:RS} in place of
Lemma~\ref{lem:split} and Remark~\ref{rem:split}. For the last part, 
choose any smooth spacelike Cauchy surface $\Sigma$ intersecting $O$
and $O'$;\footnote{The existence of such a $\Sigma$ may be seen, for example, as follows: there certainly exist smooth spacelike Cauchy surfaces $\Sigma_1$ and $\Sigma_2$ that intersect, respectively, $O$ and $O'$ in open sets. Choosing compact submanifolds-with-boundary $H_1$ (resp., $H_2$) of $\Sigma_1$ (resp., $\Sigma_2$) contained in $O\cap\Sigma_1$ (resp., $O'\cap\Sigma_2$), the union $H_1\cup H_2$ is acausal as well as being a spacelike compact submanifold-with-boundary
in $\Mb$ and the existence of $\Sigma$ follows from~\cite[Thm~1.1]{Bernal:2005qf}.} then there certainly exist open relatively compact subsets $S$ and $T$ of $\Sigma$ so that
$(S,T)$ is a regular Cauchy pair with $D_\Mb(S)\subset O\subset D_\Mb(T)$ (e.g., take $T=J_\Mb(O)\cap\Sigma$),
and we apply the first part of the result along with Remark~\ref{rem:RS}.
\end{proof}

\begin{remark}\label{rem:multiRS}\emph{
For exactly the same reason as in Remark~\ref{rem:multisplit}, Theorem~\ref{thm:RS}
may be extended to yield a state that has the Reeh--Schlieder property simultaneously
for finitely many regular Cauchy pairs specified in a common Cauchy surface of $\Mb$.}
\end{remark}

The following result reproduces the main statement of~\cite[Thm 4.1]{Sanders_ReehSchlieder}.
\begin{corollary} \label{cor:RS}
Let $\Af$ be a locally covariant QFT and assume the geometric hypotheses of Theorem~\ref{thm:RS}.
Suppose that $\omega_\Nb$ has the property that its GNS vector is cyclic
for each $\pi_{\omega_\Nb}(\Af^\kin(\Nb;O))''$ indexed by a nonempty relatively compact $O\in\OO(\Nb)$ with nontrivial causal complement $O'$. Then the conclusions of Theorem~\ref{thm:RS} hold.
\end{corollary}
\begin{proof}
We need only prove that $\omega_\Nb$ is Reeh--Schlieder for all regular
Cauchy pairs $(S_\Nb,T_\Nb)$ of $\Nb$. By hypothesis, the GNS vector $\Omega_{\omega_\Nb}$
is cyclic for $\pi_{\omega_\Nb}(\Af^\kin(\Nb;D_\Nb(S_\Nb)))''$, so we need
only prove that it is separating for $\pi_{\omega_\Nb}(\Af^\kin(\Nb;D_\Nb(T_\Nb)))''$. 
Choose any nonempty relatively compact $O\in\OO(\Nb)$ contained in the causal complement $T_\Nb'$ of $T_\Nb$ (so $O$ itself also has nontrivial causal complement $O'$ containing $T_\Nb$), whereupon  $\Omega_{\omega_\Nb}$ is cyclic for $\pi_{\omega_\Nb}(\Af^\kin(\Nb;O))''$, 
and hence separating for (any subalgebra of) 
$\pi_{\omega_\Nb}(\Af^\kin(\Nb;O))'$. By Einstein causality, 
this includes $\pi_{\omega_\Nb}(\Af^\kin(\Nb;D_\Nb(T_\Nb)))$
and its weak closure.
\end{proof}
For a theory with states $\Xf:\Loc\to\CAlgSts$, we may argue further that
the state $\nu^*\omega_\Nb$ belongs to $\Sf(\Mb)$. If one assumes
that each $\Sf(\Mb)$ is a full local-equivalence class then further
conclusions on the existence of states that are Reeh--Schlieder for 
arbitrary globally hyperbolic regions of $\Mb$ may be obtained --
see \cite{Sanders_ReehSchlieder}, which also discusses various applications of
these results. 

We have emphasized that the proofs of Theorems~\ref{thm:split}
and~\ref{thm:RS} run in close analogy. Indeed, they may be combined. 
\begin{theorem} \label{thm:RSsplit}
Assume the hypotheses of Theorem~\ref{thm:split}. If, in addition, 
$\omega_\Nb$ is Reeh--Schlieder for all regular Cauchy pairs in $\Nb$, then the state
$\nu^*\omega_\Nb$ has both the Reeh--Schlieder and split properties for $(S_\Mb,T_\Mb)$.
\end{theorem}
\begin{proof}
We combine the proofs of Theorems~\ref{thm:split} and~\ref{thm:RS}.
The value $t_*$ may be chosen so that $\{t_*\}\times\Sigma$ contains 
regular Cauchy pairs $(\sS ,\sT)$ and $(S_*,T_*)$ with
\begin{equation}\label{eq:splitRSM}
(\sS ,\sT)\prec_\Mb (S_\Mb,T_\Mb)\prec_\Mb (S_*,T_*) ,
\end{equation}
while $t_\Nb>t_*$ may be chosen so that $\{t_\Nb\}\times\Sigma$
contains regular Cauchy pairs $(S_\Nb,T_\Nb)$ and $(\leftidx{_\Nb}{}{S},\leftidx{_\Nb}{}{T})$ such that 
\begin{equation}\label{eq:splitRSN}
(\leftidx{_\Nb}{}{S},\leftidx{_\Nb}{}{T})\prec_\Nb (\sS ,\sT),\qquad
(S_*,T_*)\prec_\Nb (S_\Nb,T_\Nb).
\end{equation}
Constructing the interpolating metric as in the proof of
Theorem~\ref{thm:split}, the orderings \eqref{eq:splitRSM} and~\eqref{eq:splitRSN}
hold with $\prec_\Mb$ and $\prec_\Nb$ replaced by $\prec_\Ib$, and we may
deduce 
\begin{equation}\label{eq:splitRSI}
(\leftidx{_\Nb}{}{S},\leftidx{_\Nb}{}{T})\prec_\Ib (S_\Mb,T_\Mb)
\prec_\Ib  (S_\Nb,T_\Nb).
\end{equation}
Now $\omega_\Nb$ has the Reeh--Schlieder property for $(S_\Nb,T_\Nb)$ 
and is split for $(\leftidx{_\Nb}{}{S},\leftidx{_\Nb}{}{T})$ in $\Nb$,
and hence the same is true in $\Ib$ for $(\Af(\delta)\circ\Af(\gamma)^{-1})^*\omega_\Nb$. 
By \eqref{eq:splitRSI} and Remarks~\ref{rem:split} and~\ref{rem:RS}, 
$(\Af(\delta)\circ\Af(\gamma)^{-1})^*\omega_\Nb$ is both
Reeh--Schlieder and split for $(S_\Mb,T_\Mb)$, again as a regular Cauchy pair in $\Ib$. Hence $\nu^*\omega_\Nb$
is both Reeh--Schlieder and split for $(S_\Mb,T_\Mb)$ in $\Mb$. 
\end{proof}
\begin{remark}\label{rem:multisplitRS}\emph{%
This result also extends to the case of finitely many Cauchy pairs in a common Cauchy surface.}
\end{remark}

\subsection{Standard split inclusions and applications}

In the situation of Theorem~\ref{thm:RSsplit}, 
but now writing $(S,T)$ for $(S_\Mb,T_\Mb)$, let 
$\tilde{S}$ be an open subset of the Cauchy surface
containing $S$ and $T$ such that 
$\overline{\tilde{S}}\subset T\setminus \overline{S}$. Then
$(\tilde{S},T)$ is a regular Cauchy pair lying in a common Cauchy surface with $(S,T)$. Applying Remark~\ref{rem:multisplitRS},
the construction of $\nu$ may be arranged so that   $\omega=\nu^*\omega_\Nb$ has the Reeh--Schlieder and split properties for both $(S,T)$ and $(\tilde{S},T)$. Writing
$(\HH,\pi,\Omega)$ for the corresponding GNS representation, we define 
\begin{equation}
\Rgth_U =\pi(\Af^\kin(\Mb;D_\Mb(U)))'', 
\end{equation}
where $U$ is any of $S,\tilde{S},T$.
So far, we have $\Rgth_S\subset\Ngth\subset \Rgth_T$ and that
$\Omega$ is cyclic for $\Rgth_S$ (hence also for $\Ngth$ and $\Rgth_T$).
Moreover $\Omega$ is cyclic for $\Rgth_{\tilde{S}}$, and therefore also for
$\Rgth_T\wedge \Rgth_S'$ (using Einstein causality and causal
disjointness of $S$ and $\tilde{S}$). On the other hand,  
$\Omega$ is separating for $\Rgth_T$ and, therefore, for its subalgebras $\Rgth_S$ and
$\Rgth_T\wedge \Rgth_S'$.
In summary, the inclusion $\Rgth_S\subset \Rgth_T$ is split, and $\Omega$ is
cyclic and separating for each of $\Rgth_S$, $\Rgth_T$ and
$\Rgth_T\wedge \Rgth_S'$.\footnote{Note that the split property
for $(\tilde{S},T)$ was not used in this argument.} In the terminology of \cite{DopLon:1984},
the triple $(\Rgth_S,\Rgth_T,\Omega)$ is, therefore, a \emph{standard split inclusion}. 

Excluding a trivial situation in which $\Rgth_T=\CC\II$ (which can
only arise if the GNS space $\HH$ is one-dimensional)  it follows that
both $\Rgth_S$ and $\Rgth_T$ are properly
infinite von Neumann algebras with separable preduals, and the
Hilbert space $\HH$ is infinite-dimensional and separable~\cite[Prop.~1.6]{DopLon:1984}.\footnote{To bring out
the main point: $\Omega$ is a faithful normal state on $\Ngth$,
which is therefore countably decomposable, and hence (by
virtue of being a type $\text{I}$ factor) isomorphic to $\BB(\KK)$
where $\KK$ has countable dimension~\cite[7.6.46]{KadisonRingrose:iv}. That is, $\Ngth$ is of type $\text{I}_\infty$. As $\Omega$ is
cyclic for $\Ngth$, separability of $\HH$ follows.}
There is also a unitary $W:\HH\to\HH\otimes\HH$ with the properties
\begin{align}
W A B' W^{-1} &= A\otimes B' \qquad (A\in\Rgth_S,~B'\in\Rgth_T') \nonumber\\
W\Rgth_S W^{-1}& = \Rgth_S\otimes\II_{\HH} \nonumber\\ 
W\Rgth_T' W^{-1}& = \II_{\HH}\otimes\Rgth_T' \nonumber\\ 
W\Rgth_T W^{-1}& = \BB(\HH)\otimes\Rgth_T  
\end{align}
and we may take $\Ngth$ to be the `canonical type $\text{I}$
factor'
\begin{equation}
\Ngth= W^{-1}\left(\BB(\HH)\otimes \II_{\HH}\right) W.
\end{equation}
It is conventional to denote the split inclusion by $\Lambda=(\Rgth_S,\Rgth_T,\Omega)$.

As is well known, various consequences follow from this situation
(see, e.g., \cite[\S V.5]{Haag}).
We give some representative applications. 

\paragraph{Statistical independence} The algebras $\Rgth_S$ and $\Rgth_T'$ are
statistically independent in the $W^*$-sense,\footnote{See~\cite{Summers:1990} for discussion of
the relation between $C^*$- and $W^*$-senses of statistical independence.} because any pair of normal states $\varphi_S$ 
and $\varphi_T'$ on these algebras
with respective density matrices $\rho_S$ and $\rho_T'$  
 induces a normal state $\varphi$ with density matrix
$\rho= W^{-1} \rho_S\otimes\rho_{T}' W$ so that
\begin{equation}
\varphi(AB') = \Tr \rho AB' = \Tr \left((\rho_S A)\otimes(\rho_{T}'B')\right) =\varphi_S(A)\varphi_{T}'(B')
\end{equation}
for $A\in\Rgth_S$, $B'\in\Rgth_T'$. 

\paragraph{Strictly localized states} States of the form $\Psi = W^{-1} \psi\otimes\Omega$ ($\psi\in\HH$, $\|\psi\|=1$) may be
regarded as states strictly localized in $D_\Mb(T)$ relative to $\Omega$, because 
\begin{equation}
\ip{\Psi}{B'\Psi} = \ip{\psi\otimes\Omega}{(\II_{\HH}\otimes B')\psi\otimes\Omega} =
\ip{\Omega}{B'\Omega}
\end{equation}
for all $B'\in\Rgth_T'$. 

\paragraph{Local implementation of gauge symmetries}
In locally covariant QFT, the global gauge group of a theory
$\Af$ may be identified with the automorphism group $\Aut(\Af)$, 
the group of natural isomorphisms of $\Af$ with itself~\cite{Fewster:gauge}.

Suppose that the state $\omega_\Nb$ is gauge invariant in the sense
that $\omega_\Nb\circ\zeta_\Nb=\omega_\Nb$ for all $\zeta\in\Aut(\Af)$,
where $\zeta_\Nb$ is the component of natural transformation $\zeta$
in spacetime $\Nb$.  
Then $\omega=\nu^*\omega_\Nb$ is gauge invariant, $\omega\circ\zeta_\Mb=\omega$ for all $\zeta\in\Aut(\Af)$,
by naturality and the definition of $\nu$, so the
GNS representation carries a unitary implementation $\zeta\mapsto U(\zeta)$ of the gauge
group $\Aut(\Af)$ under which $\Omega$ is fixed.  
Then we may define 
\begin{equation}
U_\Lambda(\zeta)=W^{-1} (U(\zeta)\otimes\II_{\HH}) W,
\end{equation}
which provides a second representation of $\Aut(\Af)$, implemented by 
unitaries belonging to $\Ngth\subset\Rgth_T$, with
\begin{equation}
U_\Lambda(\zeta)AB' U_\Lambda(\zeta)^{-1} = W^{-1}\left( U(\zeta)AU(\zeta)^{-1}\otimes B'\right) W
=U(\zeta)AU(\zeta)^{-1} B' 
\end{equation}
for $A\in\Rgth_S$, $B'\in\Rgth_T'$.  
In other words, $U_\Lambda$ is a local representation of the gauge group on $\Rgth_S$, 
leaving the commutant of $\Rgth_T$ fixed. The representation is strongly continuous
(with respect to a given topology on $\Aut(\Af)$) if and only if $U$ is, and this
construction produces local generators for the gauge group and thus a local current algebra
\cite{DopLon:1983}. In principle this discussion could be developed further
to incorporate geometric symmetries of the Cauchy surface (cf.~\cite{BucDopLon:1986}) 
by modifying the construction of the interpolating spacetimes
to ensure that the isometry is preserved throughout, and starting from an invariant
state on $\Nb$. 

\paragraph{Hyperfiniteness and type $\text{III}_1$} 
Suppose $T$ can be approximated
from within by subsets $S_k\subset T$ so that each $(S_k,S_{k+1})$ is a regular Cauchy pair, $T=\bigcup_{k\in\NN} S_k$ and
\begin{equation}
\Rgth_T = 
\bigvee_{k\in\NN} \Rgth_{S_k}  .
\end{equation}
(This inner continuity would be expected if, for example, 
the von Neumann algebras are generated by a system of fields, 
cf.~\cite{BucDAnFre:1987}; alternatively, it might be imposed
as an additivity assumption.)
Then because each inclusion $\Rgth_{S_k}\subset \Rgth_{S_{k+1}}$ 
is split there is an increasing family of type $\text{I}$ factors $\Ngth_k$ so that
\begin{equation}
\Rgth_T = \bigvee_{k\in\NN} \Ngth_k
\end{equation}
and as $\HH$ is separable, $\Rgth_T$ is seen to be hyperfinite.
If, in addition, the factors appearing in the central decomposition 
of $\Rgth_T$ are known to be of type $\text{III}_1$, 
as would happen given a suitable scaling limit~\cite[Thm 16.2.18]{BaumWollen:1992} (based on \cite{Fredenhagen:1985})
then $\Rgth_T$ is isomorphic to the unique hyperfinite $\text{III}_1$
factor~\cite{Haagerup:1987} (up to a tensor product with the centre of $\Rgth_T$).  

\bigskip
Now consider the situation of a theory with states $\Xf=(\Af,\Sf):\Loc\to\CAlgSts$
obeying local quasi-equivalence, and so that $\omega_\Nb\in\Sf(\Nb)$. 
Then the state $\omega$ discussed above lies in $\Sf(\Mb)$ and the
GNS representation $(\tilde{\HH},\tilde{\pi},\tilde{\Omega})$ of any state $\tilde{\omega}\in\Sf(\Mb)$ restricts to representations of $\Af^\kin(\Mb;D_\Mb(S))$ 
and  $\Af^\kin(\Mb;D_\Mb(T))$ that are quasi-equivalent to those obtained by restricting the GNS representation of $\omega$. As already mentioned, the corresponding
von Neumann algebras $\tilde{\Rgth}_S$, $\tilde{\Rgth}_T$ are
split, though the GNS vector is not necessarily a standard vector. 
However, some elements of the discussion above hold true as
a result of the quasi-equivalence: for instance, the Hilbert space $\tilde{\HH}$ is separable (cf.\ the proof of \cite[Thm 2.4.26]{BratRob}) and $\tilde{\Rgth}_T$ contains unitaries implementing the
global gauge group on $\tilde{\Rgth}_S$, and leaving $\tilde{\Rgth}_T'$
fixed. Of course, the type of the local von Neumann algebras is preserved, because they are isomorphic.

Further applications of the split property to the issue of independence of local
algebras can be found in~\cite{Summers:1990,Summers:2009}.\footnote{Terminology 
in these references differs
in some respects from ours, which follows~\cite{DopLon:1984}; refs.~\cite{Summers:1990,Summers:2009} refer to a pair $(\Rgth_1,\Rgth_2)$ 
as split if $\Rgth_1\subset\Ngth\subset\Rgth_2'$ for some type $\text{I}$ factor $\Ngth$.} 
A weaker condition than the split property, namely \emph{intermediate factoriality},
is studied in \cite{BrFrVe03}, where various consequences are derived. 
The interpretative framework for quantum systems described
by funnels of type $\text{I}_\infty$
factors has recently been addressed in~\cite{BucSto:2014}.
Finally, we comment on the version of the split property used
in~\cite{BrFrImRe:2014} in a discussion of the tensorial structure
of locally covariant QFTs. This differs from ours in that the type $\text{I}$ von Neumann factor is
required to lie between the $C^*$-algebras $\Af^\kin(\Mb;O)$ and $\Af(\Mb)$ for 
connected $O\in\OO(\Mb)$ with compact closure, rather than between the von Neumann algebras of nested relatively compact regions in
suitable representations. An additional continuity requirement is
also imposed in~\cite{BrFrImRe:2014}. While it seems likely that 
one could at least partly address
this version of the split property with our deformation argument, we will not do this here. Alternatively one could 
investigate whether the results of~\cite{BrFrImRe:2014} hold
under the version of the split property established here.

\subsection{The distal split property}\label{sect:distal}

The \emph{distal split property} requires only that inclusions
of nested local algebras be split when the outer region is
sufficiently larger than the inner. In Minkowski space, 
for instance, the \emph{splitting distance} $d(r)$ is 
defined~\cite{DAnDopFreLon:1987} so that, for $r>0$, the inclusion
$\Rgth_{B_r}\subset \Rgth_{B_{r+d}}$ is split for all $d>d(r)$
and non-split for $d<d(r)$, 
where $B_r$ is the open ball of radius $r$ in the $t=0$
hyperplane, centred at the origin.\footnote{
One might more strongly insist on $(\Rgth_{B_r},\Rgth_{B_{r+d}},\Omega)$ being a standard
split inclusion, for cyclic and separating vector $\Omega$.} 
 An example is given 
in~\cite[Thm 4.3]{DAnDopFreLon:1987}
that consists of infinitely many independent scalar fields of masses
$m_n=(2d_0)^{-1}\log(n+1)$ and for which the splitting distance obeys  $d_0\le d(r) \le 2d_0$ for all $r>0$. 

In this subsection we generalise the notion of the distal split property and splitting distance to the curved spacetime context and show that, 
on the one hand, a weak version of the distal split property is amenable
to deformation arguments, while on the other, that for models obeying
the timeslice axiom and local quasi-equivalence, stronger versions of
the distal split property actually imply versions of the split property. 
In particular, if (any version of) the distal split property holds, then
the splitting distance for any open ball vanishes, while if \emph{uniform
distal splitting} holds then the splitting distance vanishes for all
open relatively compact sets. These results entail that models such as that of~\cite[Thm 4.3]{DAnDopFreLon:1987}, if extended to a locally covariant theory, must fail either to obey the timeslice or local quasi-equivalence axioms. 
As the repetition of the phrase `open relatively compact' would
become tedious in this subsection, we use the abbreviation \emph{orc}
instead.  

Our first result applies deformation arguments to a weak version of the distal split property. 

\begin{theorem} \label{thm:weak_distal}
Suppose $\Af$ is a locally covariant QFT.
Let $\Mb$ and $\Nb$ have oriented-diffeomorphic Cauchy surfaces, 
and suppose $\Sigma_\Nb$ is a particular smooth spacelike Cauchy surface of $\Nb$. Suppose $\omega_\Nb$ is a state on $\Af(\Nb)$ with the property that to each orc $S_\Nb\subset \Sigma_\Nb$
with nonempty exterior, there exists a regular Cauchy pair $(S_\Nb,T_\Nb)$ in $\Sigma_\Nb$ for which $\omega_\Nb$ is split. Then, for any smooth spacelike Cauchy surface $\Sigma_\Mb$ of $\Mb$ and
any orc $S_\Mb\subset \Sigma_\Mb$ with nonempty
exterior, there exist a regular Cauchy pair $(S_\Mb,T_\Mb)$ in $\Sigma_\Mb$ and an isomorphism $\nu:\Af(\Mb)\to\Af(\Nb)$ so that
$\nu^*\omega_\Nb$ is split for $(S_\Mb,T_\Mb)$.
\end{theorem}
\begin{proof}
We may assume without loss of generality that both $\Mb$ and $\Nb$ are in standard form on manifold $\RR\times\Sigma$, so that $\{0\}\times\Sigma$ corresponds to $\Sigma_\Mb$ in $\Mb$ and $\Sigma_\Nb$ in $\Nb$.  
As usual, we abuse notation by regarding $S_\Mb$ as a subset of $\Sigma$.

We may choose other orcs $S_*$, $S_\Nb$ and $\tilde{S}$ so that $S_\Mb \Subset S_* \Subset S_\Nb\Subset \tilde{S}$, 
and so that $\tilde{S}$ has nonempty exterior. By hypothesis on $\omega_\Nb$ and $\Sigma_\Nb$, we may choose an orc $\tilde{T}$
with $\tilde{S}\Subset \tilde{T}$ so that
$\omega_\Nb$ is split for $(\tilde{S},\tilde{T})_0$ in $\Nb$, and also
further orcs $T_\Nb$, $T_*$ and $T_\Mb$ so that  
\begin{equation}
S_\Mb \Subset S_* \Subset S_\Nb\Subset \tilde{S} \Subset
\tilde{T} \Subset  T_\Nb \Subset T_* \Subset T_\Mb
\end{equation}
and $T_\Mb$ has nonempty exterior.  We now claim that there exist $0<t_*<t_\Nb$ such that 
\begin{equation}
(\tilde{S},\tilde{T})_0\prec_\Nb (S_\Nb,T_\Nb)_{t_\Nb}\prec_\Nb 
(S_*,T_*)_{t_*}
\prec_\Mb (S_\Mb,T_\Mb)_0.
\end{equation}
To see this, first apply Lemma~\ref{lem:step}(a) twice to deduce that the 
left- and right-hand orderings hold for any sufficiently small $t_\Nb, t_*$;
fixing such a $t_\Nb$, a further application of Lemma~\ref{lem:step}(a)
entails that $t_*$ may be chosen close enough to $t_\Nb$ in $(0,t_\Nb)$ so that the central ordering is also valid. 
This being so, we may infer that $\omega_\Nb$ is split for $(S_\Nb,T_\Nb)$ in $\Nb$. Constructing an interpolating metric
as in the proof of Theorem~\ref{thm:split}, we obtain an
isomorphism $\nu:\Af(\Mb)\to\Af(\Nb)$ such that $\nu^*\omega_\Nb$
is split for $(S_\Mb,T_\Mb)$ in $\Mb$. 
\end{proof}

Applied to a locally covariant QFT $\Xf=(\Af,\Sf)$ with states obeying local quasiequivalence, and supposing that $\omega_\Nb\in\Sf(\Nb)$, Theorem~\ref{thm:weak_distal} implies
that every state in $\Sf(\Mb)$ is split for $(S_\Mb,T_\Mb)$, of course. 

Next, we aim to show that distal splitting is enough to deduce that
the splitting distance vanishes for certain sets, and that the full 
split property can be inferred under some circumstances. Let $\Xf=(\Af,\Sf)$ be a locally covariant QFT with states obeying local quasi-equivalence. 
To establish our notation, $\Mb=(\RR\times\RR^{n-1},\eta,\ogth,\tgth)$ will denote Minkowski spacetime, with standard inertial
coordinates $(t,x^1,\ldots,x^{n-1})$, metric $\eta=dt\otimes dt-\sum_{i=1}^{n-1} dx^i\otimes dx^i$,
and orientation and time-orientation so that $dt\wedge dx^1\wedge\cdots\wedge
dx^{n-1}$ is positively oriented and $dt$ is future-pointing. 
If $S\subset\RR^{n-1}$ then $B(S,r)$ will denote the open ball of radius $r$ about $S$ in the Euclidean metric, and as before, $(S,T)_t$ will denote 
a regular Cauchy pair $(\{t\}\times S, \{t\}\times T)$.

\begin{definition} For any orc $S$, the 
\emph{splitting distance} $d(S)\in [0,\infty]$ is defined as the infimum over all $r>0$ such that there is a state in $\Sf(\Mb)$ with the split property for
$(S,  B(S,r))_\tau$ for some $\tau\in\RR$. We say that $\Sf(\Mb)$
has the \emph{distal split property} if $d(S)<\infty$ for every orc $S$. 
If there exists $d_0>0$ such that $d(S)\le d_0$ for every orc $S$ 
then $\Sf(\Mb)$ is said to obey the \emph{uniform distal split property};
if, further, $d_0=0$, then $\Sf(\Mb)$ is said to obey the \emph{split property}.  
\end{definition} 
Owing to local quasi-equivalence, if $r>d(S)$ then there exists $\tau\in\RR$ such that every state in $\Sf(\Mb)$ 
is split for $(S, B(S,r))_\tau$; as $\Af(\psi)^*\Sf(\Mb)=\Sf(\Mb)$
for every automorphism $\psi$ of $\Mb$, this statement is also true for every $\tau\in\RR$, and  remains true if $S$ is replaced by any of its translates or rotations.  Thus $d(S)$ depends
only on the equivalence class of $S$ under orientation-preserving Euclidean isometries.  The split property as defined above means that 
every state in $\Sf(\Mb)$ is split for every regular Cauchy pair lying in a constant-time hypersurface of Minkowski space and hence (by our results of
earlier sections) for every Cauchy pair in $\Mb$.  
We note some elementary observations:
\begin{lemma} \label{lem:easydistal}
For any orc $S$, we have
\begin{equation}\label{eq:easydistal1}
d(S) \le d(B(S,r))+r 
\end{equation}
for all $r\ge 0$. Moreover, if $d(B(R)) <\infty$ for every $R>0$, where $B(R)=B(\{0\},R)$ is
the open ball of radius $R$ about the origin, then $\Sf(\Mb)$ has the distal split property and 
\begin{equation}\label{eq:easydistal2}
d(S)\le \inf_{R>\diam(S)} (R+d(B(R)))
\end{equation}
holds for every orc $S$.
\end{lemma}
\begin{proof}
In $\RR^{n-1}$ we have $B(B(S,r_1),r_2) = B(S,r_1+r_2)$. 
Considering the chain of inclusions $S\subset B(S,r)\subset B(B(S,r),\rho)= B(S,r+\rho)$ 
we easily see that $\rho>d(B(S,r))$ implies $\rho+r>d(S)$, and \eqref{eq:easydistal1} follows.

Next, suppose that all open balls have finite splitting distance. Let $S$ be any orc containing the origin and let $R>\diam(S)$, $\rho>d(B(R))$. Considering the inclusions 
$S\subset B(R)\subset B(R+\rho)\subset B(S,R+\rho)$, we see that $d(S)\le R+\rho<\infty$. As $d(S)$
is invariant under translations of $S$,  \eqref{eq:easydistal2} holds for all orcs $S$ and $\Sf(\Mb)$ obeys the distal split condition. 
\end{proof} 

A less trivial observation is the following. 
\begin{theorem} \label{thm:distal_diffeo}
Let $f\in\Diff(\RR^{n-1})$ be any diffeomorphism with uniformly bounded derivatives. For any orc $S$, $\epsilon>0$ and $r>d(B(f(S),\epsilon))$, 
one has
\begin{equation}\label{eq:diffeo_bd1}
d(S) \le \inf\{\rho>0: f^{-1}(B(f(S),r+2\epsilon))\subset B(S,\rho)\}.
\end{equation}
In particular, this gives an estimate
\begin{equation}\label{eq:diffeo_bd}
d(S)\le \kappa \,d^+(f(S)) 
\end{equation} 
where $d^+(T):=\liminf_{\epsilon\to 0+} d(B(T,\epsilon))$ is the \emph{upper splitting distance}
and $\kappa$ is the supremum of $\|D(f^{-1})\|$ over $B(f(S),r)\setminus f(S)$.
\end{theorem}
Before giving the proof, we illustrate this theorem with two examples. 
First, suppose   $f(\xb)= \xb/\lambda$, for $\lambda>0$, in which case  
$\kappa=\lambda$ and we find
\begin{equation}\label{eq:lin_scaling}
d(S)\le \lambda d^+(\lambda^{-1}S)\quad\text{and hence}\quad
d(\lambda S)\le \lambda d^+(S)
\end{equation}
for every orc $S$. Thus splitting distances scale at most linearly. Consequently, we have:
\begin{corollary}\label{cor:distal_scale}
If $d(B(R))<\infty$ for some $R>0$ then $\Sf(\Mb)$ has the distal split property.
If $\Sf(\Mb)$ has the uniform distal splitting property then $\Sf(\Mb)$ has the split property. 
\end{corollary}
\begin{proof}
If $d(B(R))$ is finite for some $R$, then by \eqref{eq:easydistal1} one sees that
$d^+(B(r))<\infty$ for any $r<R$ and hence by \eqref{eq:lin_scaling}, $d(B(r'))<\infty$
for all $r'>0$. The distal split property follows by the second part of Lemma~\ref{lem:easydistal}.

If the uniform distal split property holds then we have $d(S)\le \lambda d_0$ for all $\lambda>0$ and
any open relatively compact $S$. Thus $d(S)=0$ for all such $S$. (Clearly it would have been
enough for this conclusion that $d(\lambda S)=o(\lambda)$ as
$\lambda\to\infty$.)
\end{proof}  

For our second example, we suppose $d(B(r_*))$ is finite and nonzero for some $r_*>0$. Let $\epsilon>0$ and set $\rho_1=\frac{1}{2}d(B(r_*+\epsilon))$ (which is finite by Corollary~\ref{cor:distal_scale}). Choose $r>2\rho_1$ and $\rho_2>\frac{1}{2}r+\epsilon$. Next, choose a real-valued $\chi\in C_0^\infty(\RR^+)$ that obeys
$\chi\equiv 0$ on $[0,r_*]$,  $\inf_{\RR^+}\chi' >-1$, and
  $\chi(\rho)  = \rho-r_*$ for $\rho\in[r_*+\rho_1,r_*+\rho_2]$
(such $\chi$ certainly exist).  Then we obtain a diffeomorphism 
$f\in\Diff(\RR^{n-1})$ by 
\begin{equation}
f(\xb) = (\|\xb\|+\chi(\|\xb\|))\frac{\xb}{\|\xb\|}
\end{equation} 
which acts trivially outside a compact set and obeys
$f(B(\rho))=B(\rho+\chi(\rho))$; in particular, 
$f(B(r_*))=B(r_*)$. Applying Theorem~\ref{thm:distal_diffeo}, 
equation~\eqref{eq:diffeo_bd1} gives 
\begin{equation}
d(B(r_*))\le \inf\{\rho>0: r_*+r+2\epsilon \le r_*+\rho + \chi(r_*+\rho)\}
\end{equation}
Noting that $\rho + \chi(r_*+\rho)=2\rho$ for any $\rho\in [ \rho_1, \rho_2]$, and that this interval contains $\frac{1}{2}r+\epsilon$ in its interior, we therefore have $d(B(r_*))\le 
\frac{1}{2}r+\epsilon$, and hence
\begin{equation}\label{eq:distal_bd}
d(B(r_*)) \le \frac{1}{2} d (B(r_*+\epsilon)) + \epsilon,
\end{equation}
because $r$ was arbitrary apart from the constraint
$r>2\rho_1$. The inequality~\eqref{eq:distal_bd} holds for 
all $\epsilon>0$ and any $r_*>0$ (our argument assumed
that $d(B(r_*))>0$, but the statement holds trivially if $d(B(r_*))=0$). 
Next, take any $r>0$ and iterate \eqref{eq:distal_bd} over two subintervals
of length $\epsilon/2$, with $r_*=r$ and $r_*=r+\epsilon/2$, thus
obtaining $d(B(r)) \le \frac{1}{4} (d (B(r+\epsilon)) + 3\epsilon)$, 
also valid for all $r>0$, $\epsilon>0$. Repeating the bisection process $k$ times
in total, one finds
\begin{equation}
d(B(r)) \le \frac{d (B(r+\epsilon))}{2^{2^k}}  +\frac{(1-2^{-2^k})\epsilon}{2^{k-1}},
\end{equation}
and taking $k\to\infty$, we deduce that $d(B(r))=0$ for all $r$. 
The upshot of this argument is: 
\begin{corollary} If $d(B(r_*))<\infty$ for some $r_*>0$ then $d(B(r))=0$ for every $r>0$, 
and $d(S)=0$ for every open relatively compact $S$ that is diffeomorphic to an open ball under $f\in\Diff(\RR^{n-1})$ with bounded derivatives. 
\end{corollary}
\begin{proof}
We have already proved that $d(B(r))=0$ for all $r>0$.  
The remaining statement follows by Theorem~\ref{thm:distal_diffeo}: we may assume
that $f(S)=B(R)$ for some $R>0$ and hence $d(B(f(S),\epsilon))=d(B(R+\epsilon))=0$ for all $\epsilon>0$, so $d^+(f(S))=0$
and $d(S)=0$ by~\eqref{eq:diffeo_bd}. 
\end{proof}

This result stops slightly short of proving that the full split property holds if a ball of some radius has a finite splitting distance. The arguments used here cannot exclude the possibility that, for example, a hollow ball with inner and outer radii $a$ and $b$ might have a splitting distance $a$ (although
this would be excluded if one assumes uniform distal splitting). Of course, the interpretation of
these results is that models with nonzero splitting distances for balls, such as those of~\cite[Thm 4.3]{DAnDopFreLon:1987}, cannot be compatible with the axioms of local covariance together
with the timeslice property and local quasi-equivalence. At least heuristically one may understand the reason as follows: a nonzero splitting distance
is related to the existence of a maximum admissible temperature, which is understood technically as the statement that KMS states of any higher temperature fail to be locally quasi-equivalent  to the vacuum state~\cite{BucJun:1986}. The argument we will shortly present to prove Theorem~\ref{thm:distal_diffeo} is based on a spacetime metric
in which a period of inflation occurs between two constant time hypersurfaces, while the metric takes the Minkowski form to the past and future of this region. As inflation tends to cool temperatures,  
one might expect that some KMS states with subcritical temperature in the future of the inflation must have arisen from states with (at least locally) supercritical temperatures to the past of the inflationary period. Therefore
the evolution induced by the timeslice axiom cannot preserve the
local quasi-equivalence class.
 
\begin{proof}[Proof of Theorem~\ref{thm:distal_diffeo}]
Let $t_*=\epsilon/3$  
and let $P=(-\infty,0)\times\RR^{n-1}$, $F=(t_*,\infty)\times \RR^{n-1}$, and define
$\Pb=\Mb|_P$ and $\Fb=\Mb|_F$, with inclusion morphisms $\iota_{\Mb;P}:\Pb\to\Mb$ and $\iota_{\Mb;F}:\Fb\to\Mb$. 

As $f\in\Diff(\RR^{n-1})$ has uniformly bounded derivatives,  the push-forward by $f$ has a bounded Euclidean norm
on tangent vectors, i.e., there is a constant
$c>0$ such that $c\| f_* u\|\le \|u\|$ for every $u\in T_p\RR^{n-1}$ and $p\in\RR^{n-1}$, where $\|\cdot\|$ is the Euclidean norm on tangent spaces of $\RR^{n-1}$.  Define a metric $h$ on
$\RR^{n-1}$ so that $f^*h=\delta$, the Euclidean metric, and set
\begin{equation}
g = c^{2\varphi} \left( dt\otimes dt -  c^{-2} \varphi  h -(1-\varphi) \delta\right)
\end{equation}
where  
$\varphi\in C^\infty(\RR\times\RR^{n-1})$ takes values in $[0,1]$, with $\varphi\equiv 0$ on $F$ and $\varphi\equiv 1$ on $P$. Note that $g$ is a smooth metric, with the following properties:
\begin{tightitemize}
\item   $g|_F= \eta$, while
$g|_P=c^2 dt\otimes dt -    h$;
\item  as quadratic forms $g\le c^{2\varphi}\eta$, because $h(f_*u,f_*u)=\delta(u,u)\ge c^2\delta(f_*u,f_*u)$ for all $u$;
\item every $g$-causal curve is therefore $\eta$-causal and the Cauchy development of any
set with respect to $\eta$ is thus contained in the Cauchy development with respect to $g$;
\item every surface $\{t\}\times\RR^3$, being a Cauchy surface for $\eta$, is
a Cauchy surface for $g$, which is accordingly globally hyperbolic~\cite[Cor.~14.39]{ONeill}.
\end{tightitemize}
We now define a spacetime $\Ib=(\RR\times\RR^{n-1}, g, \ogth,\tgth)$
which is an object in $\Loc$. The map $\beta(t,\xb)= (t/c, f(\xb))$ defines a morphism $\beta:\Pb\to\Ib$, and we have a Cauchy chain
\begin{equation} \label{eq:distalCauchychain}
\Mb\xleftarrow{\iota_{\Mb;P}} \Pb \xrightarrow{\beta} \Ib \xleftarrow{\iota_{\Ib;F}} \Fb \xrightarrow{\iota_{\Mb;F}} \Mb,
\end{equation} 
where the morphisms other than $\beta$ are all induced by inclusions. 
An important consequence of our comments about Cauchy developments
with respect to $g$ and $\eta$ is that the partial ordering $\prec_\Ib$
of regular Cauchy pairs in $\Ib$ is coarser than the Minkowski ordering $\prec_\Mb$: $(S_1,T_1)\prec_\Mb (S_2,T_2)$ implies $(S_1,T_1)\prec_\Ib (S_2,T_2)$.

We now turn to our region $S$ of interest.
Letting  $\tau=-ct_*$, we have 
$\beta(\{\tau\}\times S) = \{-t_*\}\times f(S)$, and standard Minkowski geometry gives
\begin{equation}
( B(f(S),\epsilon),  B(f(S),\epsilon+r))_{2t_*}\prec_\Ib 
( f(S),  B(f(S),2\epsilon+r))_{-t_*}  
\end{equation}
for any $r>0$; for this ordering certainly holds with respect to $\prec_\Mb$,
in which we have unit speed of light and a time separation of $3t_*=\epsilon$ between the hypersurfaces containing these regular Cauchy pairs.
Applying Lemma~\ref{lem:split} and Remark~\ref{rem:split} as in the proof of Theorem~\ref{thm:split} and using local quasi-equivalence of $\Sf(\Mb)$, we may conclude that if $r>d(B(f(S),\epsilon))$, 
then every state in $\Sf(\Mb)$ is split for the regular Cauchy pair 
$( S,  f^{-1}(B(f(S),2\epsilon+r)))_{-\tau}$.  
Accordingly, we have $d(S)<\rho$ for all $\rho>0$ such that 
$B(S,\rho)$ contains $f^{-1}(B(f(S),2\epsilon+r))$, 
thus establishing \eqref{eq:diffeo_bd1}. 

To complete the proof we must estimate $\rho$. Take any $\kappa'>\kappa$, and note that $\|D(f^{-1})\|\le\kappa'$ on 
$f^{-1}(B(f(S),2\epsilon+r))\setminus f(S)$ for all sufficiently small $\epsilon>0$. Then $f^{-1}(B(f(S),2\epsilon+r))\subset B(S,\kappa'(2\epsilon+r))$, so \eqref{eq:diffeo_bd1} gives
$d(S)\le \kappa' (2\epsilon+d(B(f(S),\epsilon)))$ as $r$ may be
chosen arbitrarily close to $d(B(f(S),\epsilon))$. Taking the limit inferior
as $\epsilon\to 0+$ gives $d(S)\le \kappa' d^+(f(S))$ and hence
\eqref{eq:diffeo_bd} on taking $\kappa'\to\kappa+$. 
\end{proof}

\section{Ultrastatic spacetimes}\label{sect:ultrastatic}

In this section we comment briefly on sufficient conditions for 
a locally covariant QFT to admit a state obeying both the
split and (full) Reeh--Schlieder properties on the class of
connected ultrastatic globally hyperbolic spacetimes, i.e., those spacetimes $\Nb\in\Loc$ 
in standard form $\Nb=(\RR\times\Sigma, dt\otimes dt - h,\ogth,\tgth)$
where $h$ is a fixed complete\footnote{See, e.g.,~\cite[Prop. 5.2]{Kay1978} for the relation of completeness to global hyperbolicity.} Riemannian metric on $\Sigma$, which is
assumed connected.
As every connected spacetime $\Mb\in\Loc$ has Cauchy surfaces
oriented-diffeomorphic to those of such an ultrastatic spacetime
(by virtue of~\cite{NomizuOzeki1961} any Cauchy surface of $\Mb$ can be equipped with a complete Riemannian metric from which an ultrastatic spacetime may be constructed), such conditions
would enable the results of Section~\ref{sect:split} to 
apply nontrivially to any connected $\Mb\in\Loc$.

Let $\Nb$ be connected and ultrastatic, as defined above. Then 
$\Nb$ admits a one-parameter group of time translations $T_\tau:(t,\sigma)\mapsto (t+\tau,\sigma)$ 
and hence automorphisms $\Af(T_\tau)$ of $\Af(\Mb)$. 
Our first assumption is that $\Af(\Nb)$ admits a faithful ground state $\omega_{\Nb}$ for the time translations
$\Af(T_\tau)$. That is, (a) $\omega_{\Nb}$ is a time-translationally invariant state, 
$\Af(T_\tau)^*\omega_{\Nb}=\omega_{\Nb}$ for all $\tau\in\RR$, and 
(b) the unitary implementation $U(\tau)$ of $\Af(T_\tau)$ in the GNS
representation $(\HH_{\omega_\Nb},\pi_{\omega_\Nb},\Omega_{\omega_\Nb})$ induced by $\omega_{\Nb}$, which
obeys $U(\tau)\pi_{\omega_\Nb}(A) U(\tau)^{-1}=\pi_{\omega_\Nb}(\Af(T_\tau)A)$ and $U(\tau)\Omega_{\omega_\Nb}=\Omega_{\omega_\Nb}$, 
has a positive generator, i.e., $U(\tau)=e^{iH\tau}$ with positive self-adjoint operator $H$. In the case of a theory with states $(\Af,\Sf)$,
one would also assume that $\omega_\Nb\in\Sf(\Nb)$. 
(If $\zeta\in\Aut(\Af)$ is a global gauge transformation,
we have $\zeta_\Nb\circ\Af(T_\tau) = \Af(T_\tau)\circ\zeta_\Nb$
by naturality, and as $\zeta_\Nb$ is an isomorphism, 
$\zeta_\Nb^*\omega_\Nb\in\Sf(\Nb)$ is also a ground state.
Hence, if there is a unique ground state in $\Sf(\Nb)$, it 
is automatically gauge invariant.)

The second assumption is needed for the Reeh--Schlieder property.
Defining the local von Neumann algebras $\Rgth(O):=\pi_{\omega_\Nb}(\Af^\kin(\Nb;O))''$ for nonempty $O\in\OO(\Nb)$, we
assume the \emph{weak timelike tube criterion}
\begin{equation}
\left(\bigcup_{\tau\in\RR} \Rgth(T_\tau O)\right)'' = \Rgth(\Nb)
\end{equation} 
holds for any nonempty $O\in\OO(\Nb)$
(the right-hand side is of course $\BB(\HH_{\omega_{\Nb}})$ in
the case that $\omega_\Nb$ is pure).\footnote{E.g., this condition is fulfilled if the $\Af^\kin(\Nb;T_\tau O)$ ($\tau\in\RR$)
generate a dense subspace of $\Af(\Nb)$.}
This condition was established by Borchers in general Wightman theories in Minkowski space~\cite{Borchers:1961} and (in suitable representations) for linear fields in stationary spacetimes
by Strohmaier~\cite{Stroh:2000}.\footnote{An alternative
proof of the Reeh--Schlieder theorem on ultrastatic spacetimes, based on antilocality of fractional powers of the Laplace operator, is given
in~\cite{Verch:1993}.} Given this condition, it then holds immediately
that $\Omega$ is cyclic for every $\pi_{\omega_\Nb}(\Af^\kin(\Nb;O))$ with nonempty $O\in\OO(\Nb)$ 
and so satisfies the hypotheses of Corollary~\ref{cor:RS}. 
See, e.g., Borchers' version~\cite[Thm~1]{Borchers_vacstate:1965}
of the Reeh--Schlieder theorem~\cite{ReehSchlieder:1961}.   
It seems reasonable that the timelike tube criterion holds on \emph{connected} ultrastatic spacetimes for general theories of interest.

For the split property, we assume additionally that $\Omega_{\omega_\Nb}$ obeys a suitable \emph{nuclearity criterion}. Let $O\in\OO(\Nb)$ be nonempty and denote $\Rgth(O):=\pi_{\omega_\Nb}(\Af^\kin(\Nb;O))''$.   We say that $\omega_\Nb$ obeys the nuclearity criterion  for $O$ if the maps $\Xi_\beta:\Rgth(O)\to \HH_{\omega_\Nb}$ given for $\beta>0$ by $\Xi_\beta(A)=e^{-\beta H}A\Omega_{\omega_\Nb}$, are \emph{nuclear}. That is, for each $\beta$ there is a countable decomposition
$\Xi_\beta(\cdot) = \sum_i \psi_i\varphi_i(\cdot)$ for vectors $\psi_i\in\HH_{\omega_\Nb}$ and bounded
linear functionals $\varphi_i$ on $\Rgth(O)$ such that $\sum_i \|\psi_i\|\,\|\varphi_i\|$ is finite,
whereupon we write $\|\Xi_\beta\|_1$ for the infimum of this sum over
all possible decompositions -- a quantity called the \emph{nuclearity index}. Using~\cite[Prop.~17.1.4]{BaumWollen:1992} (which
is abstracted from~\cite{BucDAnFre:1987}), one easily sees that if
$(S,T)$ is a regular Cauchy pair in $\Nb$ and 
$\omega_\Nb$ obeys nuclearity for $D_\Nb(T)$ with the corresponding nuclearity index
obeying $\|\Xi_\beta\|_1\le e^{(\beta_0/\beta)^n}$ for some fixed $n>0$, $\beta_0>0$ 
and all $\beta\in(0,1)$, then $\omega_\Nb$ has the split property for
$(S,T)$ with $\Omega_{\omega_\Nb}$ as a cyclic and separating vector.  
In the Minkowski space theory, nuclearity conditions of this type are closely related to
good thermodynamic properties such as the existence of KMS states~\cite{BucJun:1986,BucJun:1989}, so again
there is good reason to believe that they should hold for theories of interest. In
ultrastatic spacetimes, 
nuclearity was established for the Klein--Gordon field in~\cite{Verch_nucspldua:1993}
and for Dirac fields in~\cite{DAnHol:2006}.

In summary, there is good reason to believe that physically well-behaved 
locally covariant theories should admit states satisfying the Reeh--Schlieder 
and split properties in connected ultrastatic spacetimes, and hence that
the results of Section~\ref{sect:split} apply nontrivially to yield states with the split and partial Reeh--Schlieder properties in general connected
globally hyperbolic spacetimes. 

The question of whether Reeh--Schlieder and split states can be expected in general disconnected ultrastatic spacetimes would seem to require more detailed information concerning $\Af$. Our deformation arguments work
equally well for disconnected spacetimes, however, and one can certainly  find states on disconnected spacetimes that are sufficiently entangled
across the various components that they have the Reeh--Schlieder 
property. For example, suppose $\omega_\Mb$ has the full Reeh--Schlieder property on a connected spacetime $\Mb$, and let
$O\in\OO(\Mb)$ have multiple components. Then the restriction
of $\omega_\Mb$ to $\Af(\Mb|_O)$ has the full Reeh--Schlieder property 
on this disconnected spacetime. In this situation the `behind the moon'
aspect of the Reeh--Schlieder property is brought into sharp relief:
the moon need not even be in the same spacetime component as
the experimenter!

\section{Summary}

In this paper, it has been shown that the split property and Reeh--Schlieder properties
can be established for locally covariant theories, using a common framework based on regular Cauchy pairs. The proofs of these properties
become quite streamlined and can be run simultaneously, thus implying the existence of standard split inclusions and permitting the 
results of analyses such as~\cite{DopLon:1984} to be used. Sufficient
conditions have been given for the existence of states obeying the split and Reeh--Schlieder properties in ultrastatic spacetimes, whereupon a spacetime deformation argument is used to export these properties to general
globally hyperbolic spacetimes. As a bonus, our methods also show that 
(in Minkowski space) the distal split property, in combination with the 
timeslice property and the assumption that state spaces obey local quasiequivalence, actually implies (in various specific senses) that the split condition holds.

\paragraph{Acknowledgement} I thank Klaus Fredenhagen for
asking a question about the status of the distal split property, which is
answered in Section~\ref{sect:distal}, and Ko Sanders for comments
on the text and for pointing out an error in a previous formulation of Lemma~\ref{lem:Cauchypairs}.

\end{document}